\documentclass[12pt,a4paper]{article}
\usepackage{calc}
\usepackage[all]{xy}
\usepackage[centertags]{amsmath}
\usepackage{latexsym}
\usepackage{graphicx}
\usepackage{tikz}
\usetikzlibrary{arrows,shapes,positioning}
\usetikzlibrary{decorations.pathreplacing}
\usepackage{cases}
\usepackage{amssymb}
\usepackage{amsthm}
\usepackage{mathptmx}
\usepackage{titlesec}
\usepackage{color}
\usepackage[colorlinks,citecolor=blue,linkcolor=blue]{hyperref}
\usepackage{fancyhdr}
\usepackage{epsfig,float,epstopdf,array,bm}
\usepackage{newlfont}
\usepackage[english]{babel}
\usepackage{tikz}
\usepackage{fancybox}
\usepackage{t1enc}
\usepackage{amsfonts}
\usepackage{csquotes}
\usepackage{rotating}

\usepackage[a4paper, top=1cm, bottom=2.5cm, left=2.7cm, right=2.5cm]{geometry}
\newtheorem{theo}{Theorem}[section]
\newtheorem{lem}{Lemma}[section]
\newtheorem{remark}{Remark}[section]
\newtheorem{propo}{Proposition}[section]
\newtheorem{ppte}{Property}[section]

\newtheorem{raps}{Rappel}[section]

\newtheorem{defi}{Definition}[section]

\newtheorem{con}{Consequence}[section]
\newtheorem{ase}{\underline{Assertion}}[section]

\newcounter{mnotecount}[section]

\renewcommand{\themnotecount}{\thesection.\arabic{mnotecount}}

\newcommand{\warn}[1]
{\protect{\stepcounter{mnotecount}}$^{\mbox{\footnotesize
$
\bullet$\themnotecount}}$ \marginpar{
\raggedright\tiny\em $\!\!\!\!\!\!\,\bullet$\themnotecount: {\bf
Warning:} #1} }

\author{\bf  Jean-Marc Mandeng$^{1, \dag}$ \\
{\small $^1$ Laboratory of   Mathematics,
Department of Mathematics and Computer Science,}\\
{\small  Faculty of Science, University of Douala, PO Box 24157  Douala, Cameroon}\\
{\small $^\dag$ Corresponding author: Tel. +237 693 19 04 43, Email: jeanmarcmandeng@gmail.com}}
\date{}
\title{\textsl{Modeling cholera dynamics under food Insecurity and environmental contamination: A Multi-Patch approach}} 

\begin{document}

\date{}
\maketitle

\begin{abstract}
We propose a multi-patch model of cholera transmission integrating environmental contamination, human mobility, and nutritional vulnerability. The population is stratified by food security status, and transmission occurs via human contact, bacteria in the environment and contaminated food. We derive the basic reproduction number $\mathcal{R}_0$ analyze the stability of the disease-free equilibria and show a forward bifurcation. Numerical simulations illustrate how food insecurity amplifies outbreak severity and mortality. The model highlights the role of spatial heterogeneity and socio-environmental factors in shaping cholera dynamics. Moreover, results show the impact of sinks inside starting epidemic.
\end{abstract}
\section{Introduction}

Cholera remains a persistent public health concern in many developing regions, especially where access to clean water, adequate nutrition, and basic sanitation is limited \cite{WHOCholeraData}. In Sub-Saharan Africa, \textit{Vibrio cholerae} continues to cause recurrent epidemics, leading to significant morbidity and mortality, particularly among vulnerable populations \cite{Nelson2009, Che2020}. Food insecurity and malnutrition, exacerbated by poverty and environmental stressors such as droughts or floods, may amplify the risk and severity of cholera outbreaks \cite{FAO2021, UNICEF2022}. Poor nutritional status weakens the immune system and may increase susceptibility to enteric infections, creating a vicious feedback loop between undernutrition and infectious diseases \cite{Guerrant2013, Bhutta2008}.
Environmental factors also play a central role in cholera transmission. \textit{V. cholerae} can survive and proliferate in aquatic reservoirs, with its persistence influenced by temperature, salinity, rainfall, and contamination from human activities \cite{Lipp2002,Bertuzzo2010, deMagny2008}. In such contexts, the transmission of cholera is not only the result of direct human-to-human contact but also strongly coupled to environmental and social dynamics. Consequently, modeling approaches that incorporate both environmental transmission pathways and socio-economic vulnerabilities are essential to better understand disease persistence and control (\cite{Che2020}).
While numerous models of cholera transmission have been proposed \cite{codeco2001endemic, Hartley2006, tien2010multiple}, few have integrated the spatial heterogeneity of both population vulnerability and environmental exposure across multiple communities or regions. Even fewer have addressed the compounded effects of food insecurity on cholera dynamics. Yet, these two factors malnutrition and environmental contamination can interact in ways that strongly affect both the spread and severity of the disease, particularly in decentralized or poorly connected health systems.

This work is motivated by the following question: how does food insecurity influence cholera transmission and persistence at the population level, particularly when environmental contamination and spatial connectivity are accounted for? To address this, we propose a novel multi-patch compartmental model that captures the interplay between cholera transmission, food availability, and spatial structure. The model explicitly considers two distinct classes of susceptible individuals those in food security and those in food insecurity each with different vulnerabilities to infection and mortality. It also includes compartments for acutely infected individuals, chronically infected carriers, and an environmental reservoir representing water contamination.
Here, we consider a simple extension of the classical SIR model with water compartment W  by adding a foodborne transmission to obtain the resulting \textquotedblleft SCIWR-F\textquotedblright model, allowing for both person-person, person-water-person and person-food-person transmission. 
Our goal is to provide a theoretical and numerical framework that allows us to investigate the impact of food insecurity on the burden of cholera, identify threshold parameters (such as basic reproduction numbers, $\mathcal{R}_0$) associated with disease persistence or elimination, and  explore scenarios where food-related stress may give rise to complex dynamics such as bifurcation or multiple endemic equilibria.  Numerical simulations calibrated to World bank data both malnutrition and waterborne disease burden are high used to support the theoretical results and quantify the potential health losses due to the coupling of food stress and epidemic dynamics.
In particular, we show that food insecurity, acting as a vulnerability amplifier, may increase both the force of infection and the cholera-induced mortality, thereby altering the stability landscape of the system. The presence of multiple patches further reveals the importance of spatial feedbacks, with poorly connected regions acting as sources or sinks of infection. This study highlights the critical need to address food security not only as a development goal but also as a core component of epidemic resilience planning.

The paper is organized as follows. A detailed literature of the life cycle of \textit{V. cholerae} and the structure of food dynamics are  given in Section \ref{sec:cycle}. Also we present a new model in Section \ref{sec:old} and its properties who will be extend in Section \ref{sec:Model}, based on the compartmental model and its assumptions. Section \ref{sec:Analysis} contains the mathematical analysis, including the computation of $\mathcal{R}_0$ and equilibrium properties. In Section \ref{sec:Simulation}, we present numerical simulations based on data from Littoral (Douala) and its surrounding areas. Conclusions and perspectives are given in Section \ref{sec:perspectives}.
\section{Life cycle and biological background of \textit{Vibrio cholerae}}\label{sec:cycle}

\textit{Vibrio cholerae} is a Gram-negative, facultative anaerobic bacterium that thrives in aquatic environments, especially in estuarine, brackish, and coastal waters (\cite{LizarragaPartida2009}). Its life cycle is strongly influenced by ecological and environmental conditions such as temperature, salinity, and nutrient availability \cite{Colwell2004, Nelson2009}. The bacterium alternates between two major phases: a free-living stage in the environment and a parasitic stage within human hosts \cite{Nelson2009}.
In its environmental phase, \textit{V. cholerae} is capable of surviving in both planktonic and biofilm-associated forms. It adheres to biotic surfaces such as copepods, phytoplankton, and chitinous exoskeletons of aquatic invertebrates \cite{Lipp2002, Huq1995}, facilitating persistence in nutrient-limited waters. Environmental survival is further enhanced by the transition to a viable but non culturable state under unfavorable conditions such as low temperature or nutrient deprivation (\cite{Kaper1995,Lipp2002,Oliver2005}). The bacterium can remain viable in this state for extended periods and regain infectivity when conditions become favorable.

Upon ingestion of contaminated water or food, \textit{V. cholerae} enters the human gastrointestinal tract, where it must overcome gastric acidity and colonize the small intestine \cite{Oliver2005}. This colonization involves chemotaxis, mucin penetration, and the expression of key virulence factors such as the toxin-coregulated pilus and cholera toxin (\cite{Kaper1995}). Within the host, the bacterium multiplies rapidly, leading to massive fluid loss through diarrhea, which in turn contributes to environmental recontamination \cite{Huq1995}.
Excreted bacteria are often in a hyperinfectious state for several hours after being shed, with significantly increased infectivity compared to environmental strains \cite{codeco2001endemic,Merrell2002}. This hyperinfectious phase plays a critical role in epidemic amplification. Once reintroduced into the environment, the bacteria return to their aquatic phase, completing the cycle.
The environmental to human to environment loop, modulated by seasonality, temperature, and human behavior, defines the full transmission cycle of \textit{V. cholerae}. Understanding this cycle is essential to accurately model both the environmental persistence and the outbreak dynamics of cholera (\cite{Tien2010, Merrell2002}).

\section{A proposition of new models type: SIWR-F}\label{sec:old}
In a quest to improve existing models of cholera, we undertook a revision of these model with the aim of making them more realistic. Indeed, cholera transmission is not only waterborne through the ingestion of contaminated water or contact with bodily fluids from infected individuals but also foodborne, via the interaction with food/biomass, a route that current models have yet to capture. By introducing a term of the form $\beta_WW\left(1+\eta_F(F)\right)$, we propose a novel coupling between contamination dynamics and transmission risk, bridging ingestion-based exposure and environmental pathways. Moreover, if $F=0$, we obtain classical force of infection such as \cite{Tien2010}. Then, we start by propose this basic model \eqref{eq-sys} with Flowchart given in Fig.\ref{fig-1}, variables and parameters in Tab.\ref{par1} who will be extend in section \ref{sec:Model}.
   \begin{subequations}
 \begin{align}
 \dot{S}  & = \mu N -b_W(1+\eta_F(F))WS - b_ISI-\mu S \, , \\
  \dot{I}  & =b_W(1+\eta_F(F))WS + b_ISI -(\gamma+\mu) I \, , \\
\dot{W}&= \alpha I - \xi W \,  ,\\
\dot{R}&= \gamma I - \mu R \,  ,\\
\dot{F}&= rF\left(1-\frac{F}{K}\right) - a F  \, .
          \end{align}\label{eq-sys}
      \end{subequations}
  with : $N(t) = S(t)+I(t)+R(t) := N$, and $\eta_F(F) = \cfrac{F}{K}$.
 \begin{figure}[H] 
        \begin{center}
        \begin{tikzpicture}[>=stealth, node distance=1.5cm]
  \node (S) [very thick, black, ellipse,fill=green!30, inner sep=0.4cm, draw] {S};
  
  \node (I) [right=of S, very thick, black, ellipse,fill=red!30, inner sep=0.4cm, draw] {I};
  
 \node (R) [right=of I, very thick, black, ellipse,fill=green!30,inner sep=0.4cm, draw] {R};
 
 \node (W) [below=of I, very thick, black, rectangle,fill=blue!60, rounded corners = 3pt, inner sep=0.4cm, draw] {W};

 \node (F) at (2.8,2.2) [very thick, black, rectangle,fill=yellow!60, rounded corners = 3pt, inner sep=0.4cm, draw] {F};
 
 \draw[->, thick] (S)  to node[above, yshift=0.2cm] {$\mu$} (0.7, 1);
 \draw[->, thick] (I)  to node[above, yshift=0.2cm] {$\mu$} (3.6, 1);
 \draw[->, thick] (R)  to node[above, yshift=0.2cm] {$\mu$} (6.7, 1);
 \draw[->,thick,>=latex,] (S) edge[bend left] node[midway, below] {$b_I I$} (I);
  \draw[->,dashed,>=latex,thick] (I) -- node[right] {$\alpha$} (W);
  \draw[->,dashed,>=latex,thick] (F) -- node[left] {} (I);
  \draw[,dashed,thick] (F) -- node[left] {} (6,2.2);
    \draw[->,dashed,thick] (6,2.2) -- node[left] {} (R);
  \draw[,dashed,thick] (F) -- node[left] {} (0,2.2);
    \draw[->,dashed,thick] (0,2.2) -- node[left] {} (S);
  \draw[,thick] (S) -- node[left] {} (0,-1.3);
  \draw[->,thick] (2.8,-1.3) -- node[left] {} (2.8,-0.7);
  \draw[,thick] (0,-1.3) -- node[above] {$b_WW(1+\eta(.))$} (2.8,-1.3);
  \draw[->,dashed,thick] (W) -- node[above] {} (1.4,-1.3);
  \draw[->,thick] (W) -- node[above] {$\xi$} (5,-2.9);
  \draw[->,>=latex,] (I) edge[bend left] node[midway, below] {$\gamma$} (R);       
           \draw[->] (-1.8,0) -- node[above] {$\mu N$} (S);
     \draw[->, thick] (F) to[loop above] node[left] {$rF\left(1-\frac{F}{K}\right)$} (F);
  \draw[->,thick] (2.6,1.65) -- node[left] {$a$} (2.6,1);
       \end{tikzpicture}
       \end{center}
        \caption{ Flow-chart of the model \eqref{eq-sys}}
        \label{fig-1}
       \end{figure}
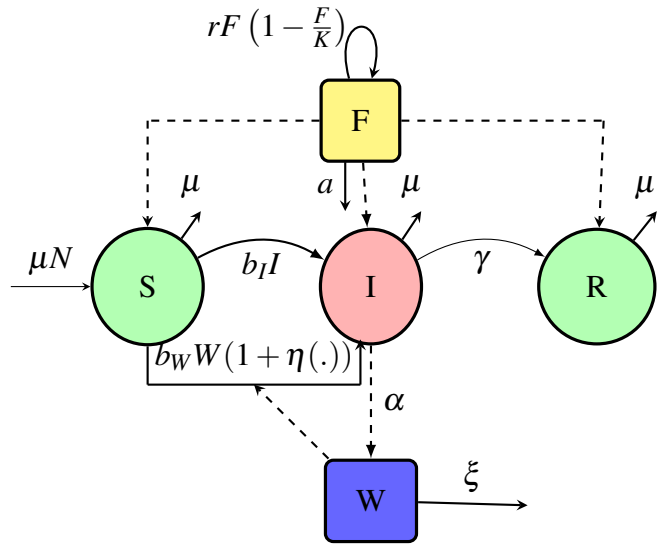     

\begin{table}[H] 
\caption{ Variables, Parameters and their biological meaning of system \eqref{eq-sys}.}
\begin{center}
\begin{tabular}{lp{4.8cm}lll}
\hline 
\textbf{Symbol} & \textbf{Biological meanings} & \textbf{Value} &\textbf{Unit} & \textbf{Source}\\ \hline 
\multicolumn{5}{l}{\textsl{Food-related parameters}\hfill}\\
$r$ & Growth rate of biomass/food & [0.1,0.5] & $Days^{-1}$ & Assumed \\
$K$ & Food holding capacity & $[2,3]\times 10^4$ &$calories^{-1}$ & \cite{UNICEF2022}\\
$a$ & Nutrition/degradation rate & [0.5,1] &$Days^{-1}$ & Assumed\\
\multicolumn{5}{l}{\textsl{Bacteria-related parameters}}\\
$\beta_W$ & Contact rate with \textit{V. cholerae} in the environment & [0.1,0.3] & $Days^{-1}$ & \cite{codeco2001endemic}\\
$\xi$ & Death rate of \textit{V. cholerae} in water reservoir & [0.1,1]& $Days^{-1}$ & \cite{Oliver2005} \\
$\alpha$ & Production of \textit{V. cholerae} by infected& $10^4-10^7$ & $cells.Days^{-1}$ & \cite{Oliver2005}\\
\multicolumn{5}{l}{\textsl{Human-related parameters}}\\
$\beta_I$ & Contact rate with \textit{V. cholerae} from human-to-human pathway & [0.05,0.15] & $Days^{-1}$ & Assumed\\ 
$\mu$ & Natural death rate of humans
 & $\frac{1}{(56\times 365)}$ & $Days^{-1}$ & WorldBank Data\\ 
 $\gamma$ & Recovery rate 
  & $[\frac{1}{7},\frac{1}{15}]$ & $Days^{-1}$ & WHO, 2022\\ 
\hline\hline
\multicolumn{5}{l}{\textbf{\textsl{Variables of system \eqref{eq-sys}}}}\\
\multicolumn{5}{l}{S(t)\hspace{1cm} Susceptible individual at time t}\\
\multicolumn{5}{l}{I(t)  \hspace{1cm} Infected individual at time t}\\
\multicolumn{5}{l}{R(t)\hspace{0.9cm} Recovered individual at time t}\\
\multicolumn{5}{l}{W(t)  \hspace{0.75cm} Pathogen concentration in water reservoir}\\
\multicolumn{5}{l}{F(t)\hspace{0.95cm} Biomass/Food density at time t}\\
\hline
\end{tabular}
\label{par1}
\end{center}
\end{table}

By giving $ \displaystyle s = \frac{S}{N}, i=\frac{I}{N}, r = \frac{R}{N}, w=\frac{\xi}{\alpha N}W$ and $f=\cfrac{F}{K}$, one has   
 \begin{align}\label{sys}
 \dot{s}  & = \mu  -\beta_W(1+f)ws - \beta_Isi-\mu s \, , \nonumber\\
\dot{i}  & =\beta_W(1+f)ws + \beta_Isi -(\gamma+\mu) i \, , \nonumber\\
\dot{w}&=  \xi(i-w) \,  ,\\
\dot{r}&= \gamma i - \mu r \,  ,\nonumber\\
\dot{f}&= \tilde{r}f(1-f) - a f \,  .\nonumber
          \end{align} 

In this study, we suppose that food is always presents in the environment, in other words, $\forall t \ge, F(t) > 0$. This condition ($r>a$) ensures that the food vector transmission component of cholera remains active at all times, which is the central focus of this work. Moreover, the case $F = 0$ reduces to the classical SIWR form already studied by \cite{tien2010multiple}.
 
   \begin{remark}
   Since the population is constant in this model, adding explicitly a term for food consumption by humans is of little interest.
   \end{remark}
      \subsection{Basic properties of the SIWR-F model}
      
Before proceeding to qualitative analysis, we first verify that the SIWR-F model is mathematically well-posed and biologically meaningful.
Since the right-hand sides of the system \eqref{sys} are composed of multivariate polynomials and hence $\mathcal{C}^\infty$, then the system is locally Lipschitz. Therefore, by the Cauchy-Lipschitz theorem, there exists a unique local solution for any initial condition in $\mathbb{R}^5_+$.

Moreover, we need to ensure that the  positive orthant is invariant under the flow of the system \eqref{sys}. Let $Z(t) = (s(t), i(t), w(t), r(t), f(t))$ denote the solution vector of the system \eqref{sys}.

\begin{lem}
Let $Z(0) \in \mathbb{R}^5_+$ with $s(0) > 0$ and $i(0), w(0), r(0), f(0) \geq 0$. Then, the solution $Z(t)$ satisfies $Z(t) \in \mathbb{R}^5_+$ for all $t \geq 0$, and $s(t) > 0$ for all $t \geq 0$.
\end{lem}

\begin{proof}
The right-hand side of each equation in system \eqref{sys} is such that if any variable reaches zero, its derivative becomes non-negative (except possibly for $i$ and $w$, which depend on infections, but are fed by $s > 0$ and the initial infection).

For $s(t)$: if $s(t_1) = 0$ at some $t_1 > 0$, then from equation (2a),
\[
\dot{s}(t_1) = \mu > 0,
\]
contradicting the assumption that $s(t)$ reaches zero. The same reasoning applies to $f(t)$ and $w(t)$. Therefore, all state variables remain nonnegative and $s(t)$ remains strictly positive.
\end{proof}

We now show that the solution is uniformly bounded for all $t \geq 0$.

\begin{lem}
The solution $Z(t)$ of system \eqref{sys} is bounded $\forall t \geq 0$ and if initial conditions satisfy $X(0) \in \Omega$, where
\begin{equation}
\Omega = \left\{ X \in \mathbb{R}_+^{5} \;\middle|\;
\begin{array}{l}
0 \leq s(t)+i(t)+r(t) \leq 1, \\
0 \leq w(t) \leq 1, \\
0 \leq f(t) \leq 1.
\end{array}
\right\}.
\end{equation}
Then, the region $\Omega$ is positively invariant under the flow of system \eqref{sys}.
\end{lem}

\begin{proof}
Since the total population is constant and the food dynamics are logistic, we can use standard comparison arguments. Indeed, $s(t) + i(t) + r(t) \leq 1$, and $f(t) \leq 1$ for the normalized model. The bacteria compartment $w(t)$ satisfies the same reasoning.
\end{proof}

These properties confirm that the SIWR-F model is well-posed and suitable for epidemiological analysis.
 \begin{defi}
 The basic reproduction rate is usually the average number of newly infected that a vector can produce in a population made up entirely  of susceptible individuals during it period of infection without any control.  
 \end{defi}
Using technique provide by \cite{Driessche2002}, we write the next generation matrix at the disease free equilibrium as $FV^{-1}$, where the $ij$ entry of the matrix $F$ is the rate at which infected individuals in compartment j produce new infections in compartment $i$, and the $jk$ entry of $V^{-1}$ is the average duration of stay in compartment j starting from k.

From \eqref{eq-sys}, we have

\begin{equation}
F = \begin{pmatrix}
\beta_I & \beta_W(1+f^*) \\
0 & 0
\end{pmatrix}, \quad \text{with} \quad f^*=1-\frac{a}{r} \quad \text{and} \quad V^{-1} = \begin{pmatrix} \frac{1}{\gamma+\mu} & 0 \\ \frac{1}{\gamma+\mu} & \frac{1}{\xi} \end{pmatrix}
\end{equation}
Then, the basic reproductive number is given by:
\begin{align}\label{eq-R0}
\mathcal{R}_0 &= \rho(FV^{-1})\nonumber\\
&= \cfrac{\beta_I+\beta_W(1+f^*)}{\gamma+\mu}.
\end{align}

To ensure that the effective control of the disease is not dependent on the initial size of \textit{V. cholerae} concentration, a global stability result must be established for the disease free equilibrium point $x^0=(1,0,0,f^*)$. 
\begin{raps}
It is easy show that when $i(t)\to 0, w(t)\to 0.$
\end{raps}
\begin{lem}\label{global}
If, $\mathcal{R}_0 \le 1$, system \eqref{eq-sys} has a disease free equilibrium who is globally asymptotically stable and unstable when $\mathcal{R}_0 > 1$.
\end{lem}
\begin{proof}
The local asymptotic stability follows clearly using \cite{Driessche2002} (Theorem 2).
Consider the following Lyapunov function using \cite{Korobeinikov2002} reasoning.
\begin{equation}
V(s,i) = s-1-ln(s)+i, \quad s,i\ge 0.
\end{equation}
Since $s\le 1, f\le f^*$ ,because the system \eqref{sys} is in proportions. Yields, 
\begin{align}
\frac{dV}{dt}\lvert_{DFE} &= \dot{s}\left(1-\frac{1}{s}\right) + \dot{i},\nonumber\\
&=  (\mu  -\beta_W(1+f)ws - \beta_Isi-\mu s)\left(1-\frac{1}{s}\right) + i(s(\beta_W(1+f) + \beta_I) -(\gamma+\mu)), \nonumber\\
&\le (\mu - \mu s)\left(1-\frac{1}{s}\right) + i(\gamma+\mu)\left(\frac{\beta_W(1+f^*) - \beta_I}{\gamma+\mu} -1\right), \nonumber\\
&\le -\frac{\mu}{s}(s-1)^2 + i(\gamma+\mu)(\mathcal{R}_0 - 1), \nonumber\\
&\le 0.
\end{align}
Since, $\frac{dV}{dt}\lvert_{DFE} = 0$ only at the disease free equilibrium, thus the largest invariant set \\$\Gamma = \left\{(s,i)\in [0,1]\times[0,1]: \frac{dV}{dt}\lvert_{DFE} = 0\right\}$ is the singleton $(1,0)$. By \cite{lasalle1976stability} invariant principle, one has that the disease free equilibrium is globally asymptotically stable in $\Gamma$. This achieves the proof.
\end{proof}
      \subsection{Stability of the endemic equilibrium}
      
 Herein, we study the stability of an endemic equilibrium whenever it exists denote $x^* = (s^*, i^*, w^*, f^*)$ of system \eqref{sys}, where
 
 \begin{equation}
 (s^*, i^*, w^*, f^*) = \left(\frac{1}{\mathcal{R}_0}, \frac{\mu}{\gamma+\mu}(1-s^*), i^*, 1-\frac{a}{r}\right).
 \end{equation}and $\mathcal{R}_0$ define in Eq.\ref{eq-R0}.
 
 We compute the Jacobian matrix \eqref{sys} evaluated at $x^*$:
 
\begin{equation}
J\arrowvert_{x^*} = 
\displaystyle \left(\begin{matrix}- \frac{\mu}{s^*} & - \beta_{I} s & -  \beta_{W} s^*(1+f^*) & - \beta_{W} s^* w^*\\\mu \left(-1 + \frac{1}{s^*}\right) & - \beta_W(1+f^*)s^* & \beta_{W} s^*(1+f^*) & \beta_{W} s^* w^*\\0 & \xi & - \xi & 0\\0 & 0 & 0 & r(1-2f^*)-a\end{matrix}\right).
\end{equation}
because, at equilibrium ($\dot{s}=0$),
\[\mu = s^*(\beta_W(1+f^*)w^* + \beta_Ii^* + \mu) \]
This implies that
\begin{equation}
\frac{\mu}{s^*} = \beta_W(1+f^*)w^* + \beta_Ii^* + \mu
\end{equation}
Moreover, for the case ($\dot{i}=0$), we have
\[\beta_W(1+f^*)s^*w^* + \beta_Is^*i^* - (\gamma + \mu)i^* = 0,\]
Using the fact that at equilibrium, $w^*=i^*$, one has
\[(\beta_W(1+f^*)s^* + \beta_Is^*)i^* = (\gamma + \mu)i^*,\]
Then
\begin{equation}
\gamma + \mu = \beta_W(1+f^*)s^* + \beta_Is^*,
\end{equation}
which implies
\begin{equation}
\beta_I s^* - (\gamma+\mu) = \beta_I s^* - \beta_W(1+f^*)s^* - \beta_Is^* = - \beta_W(1+f^*)s^*.
\end{equation}
The characteristic polynomial is $\lambda^4 + a_1\lambda^3 + a_2\lambda^2 + a_3\lambda + a_4,$ where
 \begin{align}
& a_1  = \beta_W(1+f^*)s^* + \frac{\mu}{s^*} + \xi - (-r + a) \, , \nonumber\\
&  a_2 = \beta_W(1+f^*)(\mu -(-r+a)s^*) + \mu\beta_I(1-s^*) + \frac{\mu}{s^*}(\xi-(-r+a)) -(-r+a)\xi\, , \nonumber\\
& a_3 = \mu\left[(1-s^*)\bigg(\beta_W(1+f^*)\xi + \beta_I(\xi - (-r + a))\bigg) - (-r+a)\left(\beta_W(1+f^*) + \frac{\xi}{s^*}\right)\right] \,  ,\nonumber\\
& a_4= \xi\mu(r-a)(1-s^*)(\beta_I+\beta_W(1+f^*)) \,  .
          \end{align}
Since $\mathcal{R}_0 > 1$ and $s^* = \frac{1}{\mathcal{R}_0}$, we have $ (1- s^*) > 0$, so $a_i> 0, \forall i=1,\ldots,4$. The Routh-Hurwitz criteria give that the endemic equilibrium is stable if $a_1, a_4 > 0$ and $a_1a_2-a_3>0, a_1a_2a_3-(a_1)^2a_4-(a_3)^2>0$. Clearly, we prove that $a_1a_2-a_3>0$. Indeed: 
\begin{align}
a_1a_2 - a_3 &= a_2\left(\beta_W(1+f^*)s^* + \frac{\mu}{s^*}\right) +a_2(\xi - (-r + a))-a_3 ,\nonumber\\
&= a_2\left(\beta_W(1+f^*)s^* + \frac{\mu}{s^*}\right) + \beta_{W} \mu \xi(1+f^*)s^*  + (r-a)^{2} \left(\beta_{W}(1+f^*)s^* + \frac{\mu}{s^*} + \xi\right) \nonumber\\
& + \xi(r-a) \left(\beta_{W}(1+f^*)s^*  + \frac{\mu}{s^*} + \xi\right)+ \frac{\mu \xi^{2}}{s^*} ,\nonumber\\
& > 0.
\end{align}
Thus, local stability of the endemic equilibrium is determined by the sign of the last condition
\begin{align}
C &= a_1(a_2a_3-a_1a_4)-a_3^2,\nonumber\\
&= a_1\bigg((1-s^*)\xi\mu(r-a) \frac{\left(\beta_{W}(1+f^*)s^{*2}+\mu+ s^*(\xi -  (-r+a))\right) \left(\beta_I+\beta_{W}(1+f^*)\right)}{s^*} + \nonumber\\
&( \beta_{W}(1+f^*)\mu  - \beta_{W}(1+f^*)s^{*} (-r+a) +\beta_{I} \mu(1-  s^{*})  - \xi (-r+a) +\frac{\mu (\xi-(-r+a)) }{s^*}a_3\bigg) \nonumber\\
& - a_3^2 > 0.
\end{align}
Thus, the endemic equilibrium in system \eqref{sys} is locally stable whenever it exists and we prove that it persists inside the population when it appears using \cite{thieme1992persistence} in Theorem \ref{th:persistence}.

\begin{theo}\label{th:persistence}
If $\mathcal{R}_0>1$, then the disease is {\em uniformly persistent strong} in the following sense: there exists \(\varepsilon>0\) such that for every solution with \(i(0)>0\),
\[
\liminf_{t\to\infty} i(t)\ge \varepsilon.
\]
\end{theo}

\begin{proof}
The proof is in Appendix {\color{blue} A}.
\end{proof}

 \begin{figure}[H]
 \centering
 \includegraphics[width=0.502\textwidth]{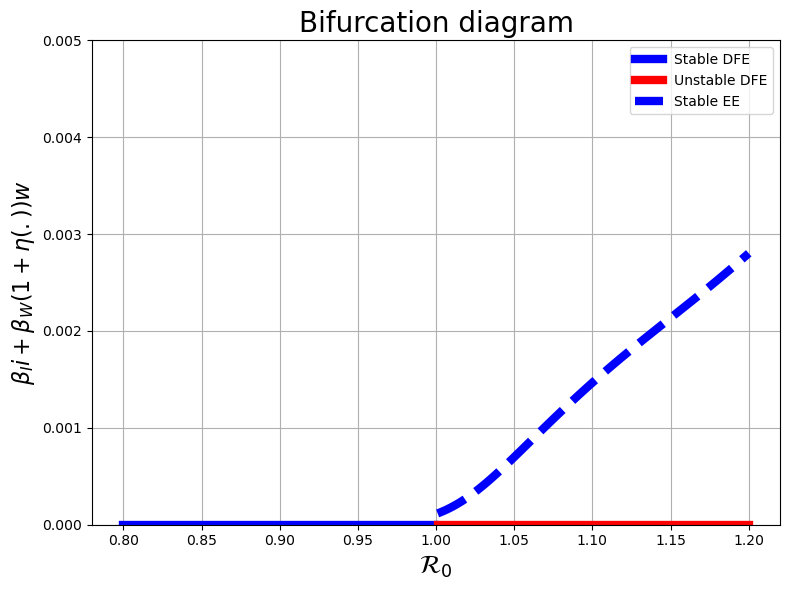}
 \includegraphics[width=0.49\textwidth]{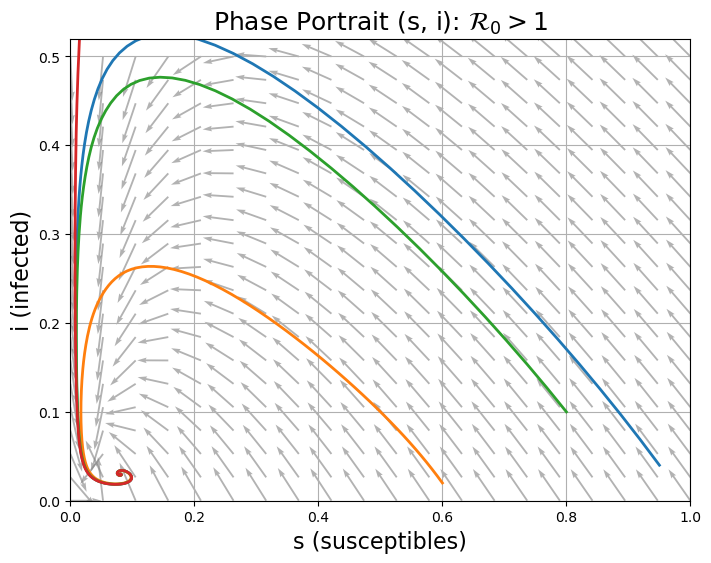}
 \caption{Bifurcation Diagram at left and Phase Portrait at right of system \eqref{eq-sys}.}
  \label{bifurcation}
 \end{figure}  

The left panel in Fig.\ref{bifurcation} displays a forward bifurcation: as the basic reproduction number  $\mathcal{R}_0$ crosses the threshold of 1, the disease-free equilibrium loses stability and an endemic equilibrium emerges. This confirms the analytical result (Lemma \ref{global}) that  is a critical threshold for cholera persistence. The right panel presents the phase portrait of the system, showing the trajectories in the plane. It demonstrates that trajectories starting from different initial conditions converge to the endemic equilibrium (if $\mathcal{R}_0 > 1$). The vector field highlights the direction of movement, reinforcing the local and global stability properties derived in the present section.


  \subsection{Comparison of SIR, SIWR and SIWR-F dynamics}
  Fig \ref{calibration} compares the yearly proportion of infected individuals reported in Cameroon from 1976 to 2000 (\cite{WHOCholeraData} data, normalized) with simulations from three models: the classical SIR (cf. \cite{AndersonMay1991}), the SIRW (cf. \cite{Tien2010}), and the proposed SIRW-F model incorporating both environmental and foodborne pathways. Calibration was performed using the parameters in Table~\ref{par1}, and the root mean square error (RMSE) was computed for each model.
  \begin{figure}[H]
  \centering
  \includegraphics[width=\textwidth]{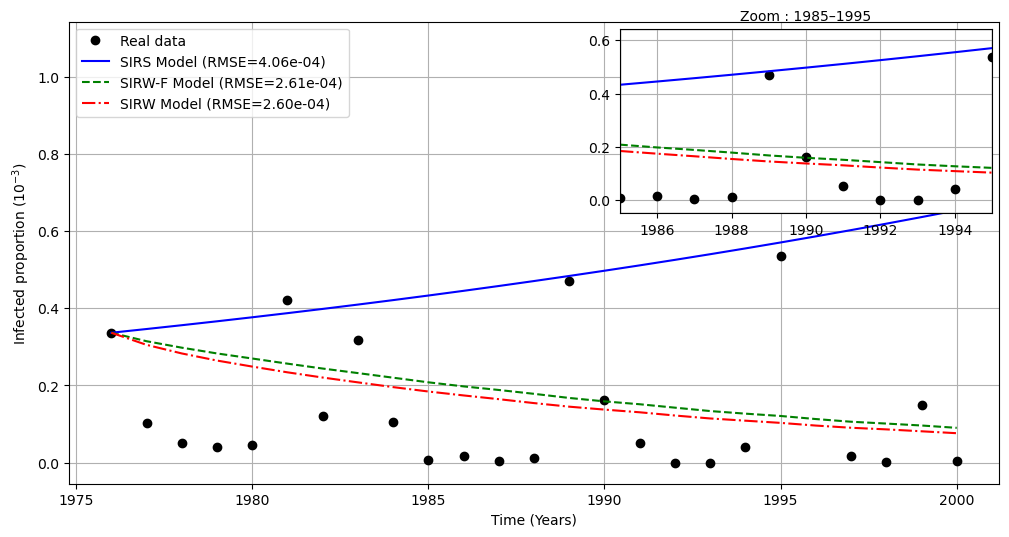}
  \caption{Comparison between SIR, SIWR and SIWR-F under cholera reported cases in Cameroon during 1976-2000 according to \cite{WorldBankData} and \cite{WHOCholeraData} Data.}
   \label{calibration}
  \end{figure}   

Although the SIRS model shows a slightly lower RMSE ($\approx 4e-4$) compared to the environmental models (SIWR and SIWR-F), it predicts an unrealistically fast increase of the disease ($\mathcal{R}_0 = 1.077>1$). In contrast, both environmental models reproduce the persistent low-level transmission observed in the data, with SIWR-F ($\mathcal{R}_0 = 0.78<1$) showing a slightly better fit in the 1985-1995 zoomed period. This improvement stems from the inclusion of foodborne transmission, which sustains the infection at a realistic endemic level and matches the long-term persistence pattern.
Moreover, the close RMSE values between SIWR (($\mathcal{R}_0 = 0.808<1$)) and SIWR-F indicate that adding the foodborne pathway does not compromise the overall statistical fit, while enhancing the model biological realism and policy relevance. These results support the use of environmentally coupled models, particularly SIWR-F, for understanding and predicting cholera dynamics in endemic settings like Cameroon.

  \begin{figure}[H]
  \centering
  \includegraphics[width=0.49\textwidth]{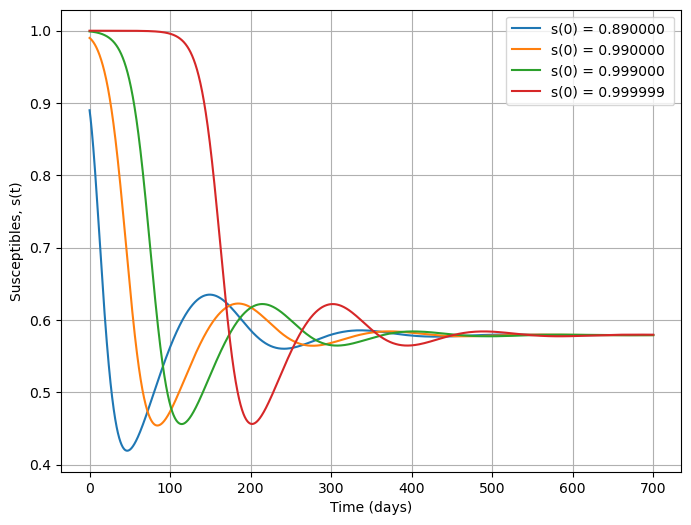}
  \includegraphics[width=0.49\textwidth]{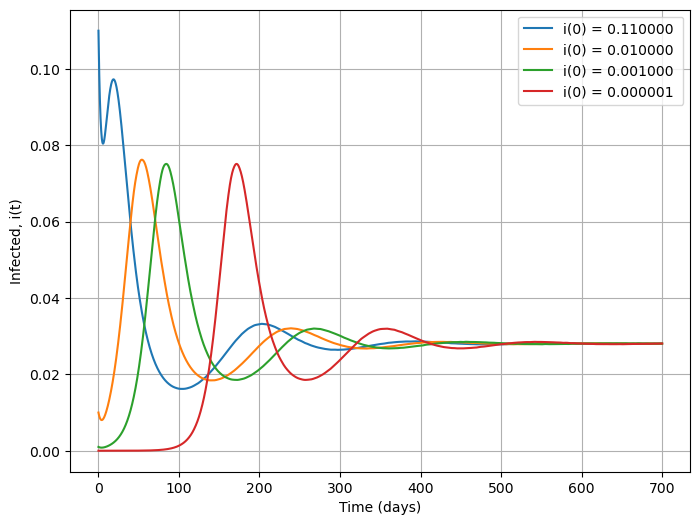}
  \label{infecte}
  \caption{Simulation of system \eqref{sys} using various initial conditions when $\mathcal{R}_0 = 1.726 > 1$. We present dynamics of susceptible and infected during an epidemic.}
  \end{figure}  
   
 \section{An extension : Multi-patch framework}\label{sec:Model}
 
However, this new modifications inside system \eqref{eq-sys} to correct the previous models (\cite{AndersonMay1991}, \cite{Eisenberg2013}, \cite{Bertuzzo2010}, \cite{Tien2010}, \cite{Che2020}...). It do not taking into account some others realities like immunity losses, mortality due to the disease, migration, presence of asymptomatic carriers, variability size of the population, non-linearity of diseases transmissions, and vulnerability of susceptible due to poverty, under-nourrished or others. So, the goal of this section is to solve them.

\subsection{Model Formulation}
We need the following hypotheses:
\begin{ase}

 \begin{itemize} \item[\color{white}]
 \item The foodborne transmission channel mediates environmental contamination by bacteria (\textit{V. cholerae}), thus capturing both indirect environmental and ingestion-based exposure routes. Then, 
 $\lambda(.) = \beta_B\cdot\cfrac{F^c_i}{\Delta_1 + F^c_i} + \beta_H\cdot\cfrac{I_i + \varepsilon C_i}{N_i^{q_p}}, \;$ where $\lambda(.)$ is the force of infection with $\eta_F = \beta_F\cdot\cfrac{F_i}{\Delta_2 + F_i}$ and $q_p \in \{0, 1\}$ i.e inside the principal node, the transmission look the standard incidence function ($q_p = 1$) and the mass action law for their satellites ($q_p=0$).
 \item $ \phi_i = \left\{ \begin{array}{ll} 
 \max (0, \delta(K-ea_iF_i))& \mbox{i = 1}\\ 
 \phi_2 & \mbox{i = 2} \end{array} 
 \right. $. For the transition to vulnerable individuals, it is assumed that vulnerability appears when the amount of nutriments absorbed (converted into biomass) is below than a threshold K (2000 calories.day$^{-1}$), i.e $ea_iF_i<K$, or even $K-ea_iF_i>0$. Individuals become vulnerable at rate $\phi_i$, with $\delta$ the apparition rate of vulnerability.
 \end{itemize} 
\end{ase}
Here, we develop a mathematical model that describes the multi-patch dynamics of cholera transmission in a population subject to nutritional vulnerability and contaminated food. The model is formulated using a compartmental framework, and explicitly considers both spatial heterogeneity (through patch structure) and multiple transmission routes (direct and indirect). Individuals move between compartments according to their epidemiological status and are exposed to infection either via direct contact with infectious individuals or through ingestion of contaminated food.

The model rests on the following assumptions:
\begin{itemize}

\item[(i)] Individuals who are food insecure or nutritionally vulnerable are more likely to become infected and suffer severe outcomes if exposed to cholera.
\item[(ii)] Individuals within a given patch are homogeneously mixed, but inter-patch coupling exists via human mobility and environmental contamination.
\item[(iii)] Food contamination arises from bacterial load in the environment and is modulated by local food availability and hygiene conditions.
\item[(v)] We assume that food and water ingestion occur jointly during meals, as is common in many societies worldwide. Therefore, waterborne and foodborne exposures are combined into a single ingestion-based transmission route. This simplifies the model compared to classical cholera frameworks where water is treated separately.
\item[(vi)] Food contamination is assumed to result primarily from environmental exposure to \textit{V. cholerae}, reflecting the dominant route of contamination observed in cholera-endemic settings. Direct human contamination (e.g., via food handling) is not explicitly modeled, as its contribution is generally secondary compared to water-related pathways.
\end{itemize}

Each spatial unit (or patch) contains human and environmental compartments. At time , individuals in patch  are divided into the following disjoint classes: the Susceptible individuals with adequate nutrition ($S_{1,i}(t)$), the Susceptible individuals under nutritional stress or vulnerable ($S_{2,i}(t)$), Symptomatic infectious individuals ($I_i(t)$) , Asymptomatic carriers ($C_i(t)$) and the Recovered individuals with temporary immunity ($R_i(t)$).
Environmental compartments include: Concentration of \textit{Vibrio cholerae} in the local environment $\left(B_i(t)\right)$, Safe food availability ($F_i^s(t)$) and Contaminated food ($F_i^c(t)$).
Individuals can move from $S_1$ to $S_2$ at rate $\phi_1$ when their nutritional intake falls below a threshold $K$ and return to $S_1$ at rate $\phi_2$. 

The total population in patch  at time  is:
\begin{equation}
N_i(t) = S_{1,i}(t) + S_{2,i}(t) + I_i(t) + C_i(t) + R_i(t).
\end{equation}

The dynamics of cholera in each patch are governed by both local processes (e.g., infection, recovery, bacterial shedding) and inter-patch interactions (e.g., movement of individuals, environmental exchange). The infection force  combines two mechanisms:

\begin{enumerate}
\item Direct human-to-human transmission, proportional to the prevalence of infectious individuals.

\item Food-borne transmission through ingestion of contaminated food, governed by a saturation function of disease.

\end{enumerate}

Safe food availability  is regenerated logistically and consumed by the population. The contamination of food depends on environmental bacterial concentration  and decreases with better hygiene at rate $\omega$.
The concentration of bacteria  increases with shedding from both symptomatic and asymptomatic individuals and decays naturally or via dilution/migration who is temperature-dependent at rate $\mu_B(T)$. Contaminated food is produced at rate proportional to the product of bacteria  and available food , and decays at a rate depending on sanitation and population size.
Migration between patches is modeled via matrices that govern movement of susceptible, infected, or bacteria, allowing the disease to propagate spatially.
The diagram of Figure \ref{fig} shows the overall flow of infection, food contamination, and vulnerability dynamics.
 Consider the following Flowchart \ref{fig} with:

 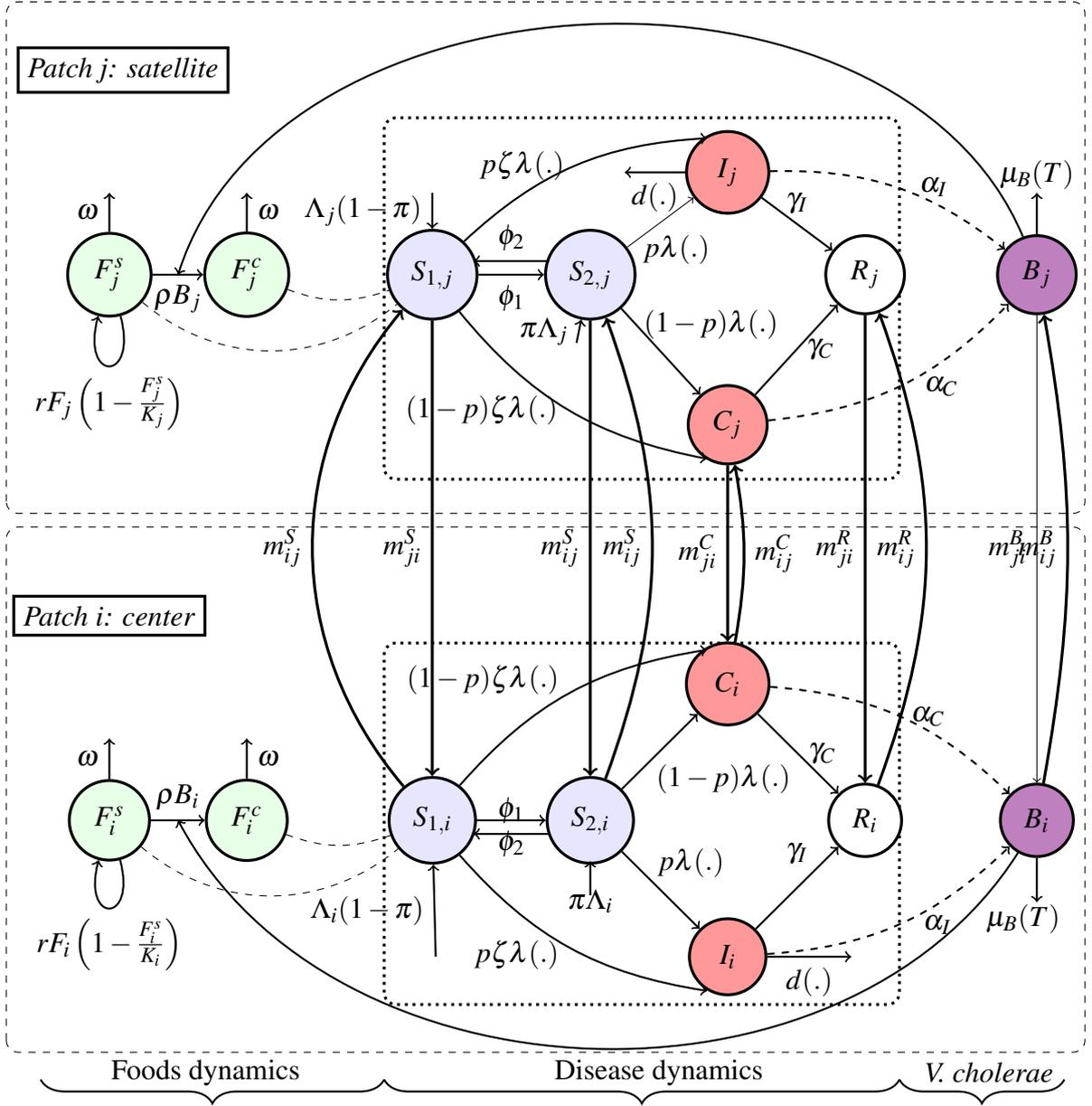
\begin{figure}[H]
 \begin{center} 
 \begin{tikzpicture}
 \draw[rounded corners=5pt, dashed] (-4.5, 3.5) rectangle (11.2, 11); 
 \draw[rounded corners=5pt, dashed] (-4.5, -4.4) rectangle (11.2, 3.3); 
 \node[very thick, black, draw] (A) at (-3, 2) {\textit{Patch i: center}};
 \node[very thick, black, draw] (A1) at (-2.8, 10) {\textit{Patch j: satellite}};

\draw[rounded corners=5pt, dotted, very thick] (1, 4) rectangle (8.5, 9.3);
\node[very thick, black, circle,fill=blue!10, inner sep=0.2cm, draw] (S) at (1.7, 7) {$S_{1,j}$};
\node[very thick, black, circle,fill=blue!10, inner sep=0.2cm, draw] (P) at (4, 7) {$S_{2,j}$};
\node[very thick, black, circle,fill=green!10, inner sep=0.2cm, draw] (F) at (-1, 7) {$F^c_j$};
\node[very thick, black, circle,fill=green!10, inner sep=0.2cm, draw] (Fs) at (-3, 7) {$F^s_j$};      
\node[very thick, black, circle,fill=red!40, inner sep=0.25cm, draw] (I) at (6, 8.5) {$I_j$};
\node[very thick, black, circle,fill=white!90, inner sep=0.2cm, draw] (X) at (8, 7) {$R_j$};
\node[very thick, black, circle,fill=red!40, inner sep=0.2cm, draw] (C) at (6, 4.8) {$C_j$};
\node[very thick, black, circle,fill=violet!50, inner sep=0.2cm, draw] (B1) at (10.5, 7) {$B_j$};

\draw[->, thick] (S) to node[below] {$\phi_1$} (P);
\draw[->, thick] (B1) to node[above,yshift=0.2cm] {$\mu_B(T)$} (10.5, 8.2);
\draw[->, thick] (3.4, 7.2) to node[above] {$\phi_2$} (2.3, 7.2);
\draw[->,thick] (P) to node[right,xshift=-0.4cm,yshift=0.4cm] {$(1-p)\lambda (.)$} (C); 
\draw[->,] (P) to node[right,xshift=-0.4cm,yshift=-0.4cm] {$p\lambda (.)$} (I);
\draw[->, thick] (I) to node[below] {$d(.)$} (4.5, 8.5);
\draw[->, thick] (I) to node[above] {$\gamma_I$} (X);
\draw[->, thick] (C) to node[right] {$\gamma_C$} (X);
\draw[->, dashed, thick] (C) to[bend right=20] node[right, xshift=0.2cm] {$\alpha_C$} (B1);
\draw[->, dashed, thick] (I) to[bend left=20] node[right, xshift=0.2cm] {$\alpha_I$} (B1);
\draw[->, thick] (S) to[bend left=20] node[left] {$p\zeta\lambda (.)$} (5.7, 9);
\draw[->, thick] (S) to[bend right=20] node[left] {$(1-p)\zeta\lambda (.)$} (5.7, 4.3);
                 
\draw[-, dashed,] (Fs) to[bend right=40] node[above] {} (S);
\draw[-, dashed,] (F) to[bend right=20] node[above] {} (S);     
\draw[->, thick] (Fs) to[loop below] node[below] {$rF_j\left(1-\frac{F_j^s}{K_j}\right)$} (F1); 
\draw[->, thick] (Fs) to node[left] {$\omega$} (-3, 8.2);  
\draw[->, thick] (F) to node[right] {$\omega$} (-1, 8.2);
\draw[->, thick] (Fs) to node[below] {$\rho B_j$} (F);
\draw[->, thick] (B1) to[bend right=70] node[below] {} (-2,7);     

      \draw[rounded corners=5pt, dotted, very thick] (1, -3.7) rectangle (8.5, 1.6);
      \node[very thick, black, circle,fill=blue!10, inner sep=0.2cm, draw] (S2) at (1.7, -1) {$S_{1,i}$};
      \node[very thick, black, circle,fill=blue!10, inner sep=0.2cm, draw] (P2) at (4, -1) {$S_{2,i}$};
      \node[very thick, black, circle,fill=green!10, inner sep=0.2cm, draw] (F2) at (-1, -1) {$F^c_i$};
      \node[very thick, black, circle,fill=green!10, inner sep=0.2cm, draw] (F1) at (-3, -1) {$F^s_i$};
     \node[very thick, black, circle,fill=red!40, inner sep=0.25cm, draw] (S3) at (6, 1) {$C_i$};
     \node[very thick, black, circle,fill=white!90, inner sep=0.2cm, draw] (X3) at (8, -1) {$R_i$};
     \node[very thick, black, circle,fill=red!40, inner sep=0.25cm, draw] (P3) at (6, -3) {$I_i$};
     \node[very thick, black, circle,fill=violet!50, inner sep=0.2cm, draw] (B) at (10.5, -1) {$B_i$};

\draw[->,very thick] (S2) to[bend left=40] node[left] {$m^S_{ij}$} (S);
\draw[->,very thick] (S) to node[left] {$m^S_{ji}$} (S2);
\draw[->,very thick] (P2) to[bend right=20] node[left] {$m^S_{ij}$} (P);  
\draw[->,very thick] (P) to node[left] {$m^S_{ij}$} (P2); 
\draw[->,] (B1) to node[left] {$m^B_{ji}$} (B); 
\draw[->,very thick] (B) to[bend right=10] node[left] {$m^B_{ij}$} (B1); 
\draw[->,very thick] (S3) to[bend right=10] node[right] {$m^C_{ij}$} (C);
\draw[->,very thick] (C) to node[left] {$m^C_{ji}$} (S3);
\draw[->,very thick] (X) to node[left] {$m^R_{ji}$} (X3);
\draw[->,very thick] (X3) to[bend right=20] node[left] {$m^R_{ij}$} (X);

   \draw[-, dashed,] (F1) to[bend right=40] node[above] {} (S2);
   \draw[-, dashed,] (F2) to[bend right=20] node[above] {} (S2);     
   \draw[->, thick] (F1) to[loop below] node[below] {$rF_i\left(1-\frac{F_i^s}{K_i}\right)$} (F1); 
   \draw[->, thick] (F1) to node[left] {$\omega$} (-3, 0.2);  
   \draw[->, thick] (F2) to node[right] {$\omega$} (-1, 0.2);
          
   \draw[->, thick] (F1) to node[above] {$\rho B_i$} (F2);
   \draw[->, thick] (B) to[bend left=60] node[below] {} (-2,-1);
\draw[decorate,decoration={brace,amplitude=10pt,mirror}, thick,black] (-4,-4.8) -- (1,-4.8);
\node at (-1.6, -4.7) {Foods dynamics};
\draw[decorate,decoration={brace,amplitude=10pt,mirror}, thick,black] (1,-4.8) -- (8.5,-4.8);
\node at (5, -4.7) {Disease dynamics};
\draw[decorate,decoration={brace,amplitude=10pt,mirror}, thick,black] (8.5,-4.8) -- (10.9,-4.8);
\node at (9.8, -4.7) {\textit{V. cholerae}};

     \draw[->, thick] (S2) to node[below] {$\phi_2$} (P2);
     \draw[->, thick] (B) to node[below,xshift=-0.2cm,yshift=-0.2cm] {$\mu_B(T)$} (10.5, -2.2);
     \draw[->, thick] (3.4, -1.2) to node[above] {$\phi_1$} (2.3, -1.2);
     \draw[->, thick] (1.7, 8.2) to node[left] {$\Lambda_j(1-\pi)$} (S);
     \draw[->, thick] (3.8, 6) to node[left] {$\pi\Lambda_j$} (P);
     \draw[->, thick] (1.75, -3) to node[left] {$\Lambda_i(1-\pi)$} (S2);
     \draw[->, thick] (4, -2.1) to node[below] {$\pi\Lambda_i$} (4, -1.6);
     \draw[->, thick] (P3) to node[below] {$d(.)$} (7.8, -3);
   \draw[->, thick] (P3) to node[above, yshift=0.2cm] {$\gamma_I$} (X3);
       \draw[->, thick] (S3) to node[right] {$\gamma_C$} (X3);
       \draw[->, dashed, thick] (S3) to[bend left=20] node[right] {$\alpha_C$} (B);
       \draw[->, dashed, thick] (P3) to[bend right=20] node[right, xshift=0.2cm] {$\alpha_I$} (B);
       \draw[->, thick] (S2) to[bend right=20] node[left, yshift=-0.1cm] {$p\zeta\lambda (.)$} (5.7, -3.5);
       \draw[->,thick] (P2) to node[right, xshift=-0.2cm, yshift=0.4cm] {$p\lambda (.)$} (P3);
       \draw[->, thick] (S2) to[bend left=20] node[left, yshift=0.2cm] {$(1-p)\zeta\lambda (.)$} (5.7, 1.5);
       \draw[->,thick] (P2) to node[right, xshift=-0.2cm, yshift=-0.4cm] {$(1-p)\lambda (.)$} (S3);  
 \end{tikzpicture}
 \caption{Schematic representation of the human-bacteria-foods dynamics during epidemic phase. For simplicity, natural
 deaths and recovery losses are not included here in the diagram but are taken into account in the model.}
 \end{center}
 \label{fig}
 \end{figure}

From the flowchart of diagram in Fig. \ref{fig}, the dynamics of the interaction between foods, the concentration of \textit{V. cholerae} in the environment and human disease transmission is given by the
following system of ordinary differential equations describing the temporal evolution of each class across all patches:

    \begin{subequations}
 \begin{align}
 \dot{S_{1,i}}  & = (1-\pi)\Lambda_i + \xi_1 R_i + \phi_2S_{2,i}  - \zeta\lambda(.)S_{1,i} - (\phi_1 + \mu) S_{1,i} + \sum_{j=1}^{n} (m_{ji}^S S_{1,j}-m_{ij}^S S_{1,i}) \, , \\
\dot{S_{2,i}}&=\pi\Lambda_i + \xi_2 R_i + \phi_1S_{1,i} - \lambda(.)S_{2,i} -(\phi_2 + \mu) S_{2,i}+ \sum_{j=1}^{n} (m_{ji}^S S_{2,j}-m_{ij}^S S_{2,i}) \, , \\
\dot{C_i}&= (1-p)\lambda(.)(S_{2,i} + \zeta S_{1,i}) - (\gamma_C + \mu) C_i+ \sum_{j=1}^{n} (m_{ji}^C C_j-m_{ij}^C C_i) \, , \\ 
\dot{I_i}&=p\lambda(.)(S_{2,i} + \zeta S_{1,i}) - (\gamma_I + d(.)+ \mu) I_i+ \sum_{j=1}^{n} (m_{ji}^I I_j-m_{ij}^I I_i) \, , \\
\dot{R_i}&=\gamma_C C_i +\gamma_I I_i - (\mu + \xi_1 +\xi_2)R_i+ \sum_{j=1}^{n} (m_{ji}^R R_j-m_{ij}^R R_i) \, , \\
\dot{F_i}^s&=r_iF_i\left(1-\cfrac{F_i^s}{K_1}\right) -(a_i+\omega) F_i^s- \rho B_iF_{i}^s \, , \\
\dot{F_{i}^c} &= \rho B_iF_{i}^s - (\omega + a_i)F_{i}^c \, ,\\
\dot{B_{i}}&= \alpha_C C_i + \alpha_I I_i - \sigma\rho B_iF_{i}^s-\mu_{B_i}(T)B_{i}+ \sum_{j=1}^{n} (m_{ji}^B B_j-m_{ij}^B B_i) \,  .
          \end{align}
 \label{systeme}
      \end{subequations}
      With initial conditions $S_{1}(0)>0, S_{2}(0)>0, C(0)>0, I(0)>0, B(0)>0, F^s(0)>0$ and $F^c(0)>0$. 
Tables \ref{var} and \ref{par} present the variables and parameters of system (\ref{systeme}), respectively. 
With $S_i=S_{1,i}+S_{2,i}$, $F_i = F_i^s + F_i^c$ and $N_i = S_{1,i}+S_{2,i}+C_i+I_i+R_i$.
\begin{table}[H]
\caption{ Variables of the SCIRW-F model \eqref{systeme}, where $i = 1, \ldots, n$ is the the index of patch.}
\begin{center}
\begin{tabular}{lll}
\hline  
\textbf{Symbols} & \textbf{Biological meanings} & \textbf{Unit}\\ \hline 
$S_{1,i}(t)$ & Susceptible without food insecurity & Number\\ 
$S_{2,i}(t)$ & Susceptible in food insecurity & Number\\ 
$C_i(t)$ & Number of asymptomatic & Number\\ 
$I_i(t)$ & Number of infected & Number\\ 
$R_i(t)$ & Recovered individuals & number\\ 
$F_{i}^s$ & Available healthy foods & calories\\ 
$F_{i}^c$ & Contaminated foods & calories\\ 
$B_{i}(t)$ & Number of bacteria in environment & cells.ml$^{-1}$\\
 \hline
\end{tabular}
\label{var}
\end{center}
\end{table}

\begin{sidewaystable} 
\caption{ Parameters and their biological meaning of system (\ref{systeme}).}
\begin{center}
\begin{tabular}{lllll}
\hline 
\textbf{Symbols} & \textbf{Biological meanings} & \textbf{Value} &\textbf{Unit} & \textbf{Source}\\ \hline 
$r_i$ & Growth rate of food-biomass & 0.1-0.5 & $Days^{-1}$ & Assumed\\
$K$ & Critical limit of food consumption& 2000 &$calories.Days^{-1}$ & \cite{FAO2021}\\
$K_1$ & Food holding capacity & 2500-3000 &$calories.Days^{-1}$ & Assumed\\
$\beta_B$ & Contact rate with \textit{V. cholerae} in the environment & 0.1-0.3 & $Days^{-1}$ & \cite{codeco2001endemic}\\ 
$\beta_H$ & Contact rate with \textit{V. cholerae} from human-to-human pathway & 0.05-0.15 & $Days^{-1}$ & Assumed\\ 
$\Delta_1$ & Half saturation rate \textit{V. cholerae}& $10^6-10^7$& $cells.mL^{-1}$ & Assumed\\ 
$\Delta_2$ & Half saturation rate for foods& $10^3-10^4$& $calories.day^{-1}$ & Assumed\\ 
$d(.)$ & Mortality rate due to disease & 0.02-0.03& $Days^{-1}$ & WHO, 2022\\
$\phi_2$ & Transition rate to vulnerable individuals & 0.1-1 & $Days^{-1}$ & Assumed \\
$\mu$ & Natural mortality rate  & $1/(55\times 365)$& $Days^{-1}$ & World Bank\\
$\mu_B$ & Mortality rate for \textit{V. cholerae} & 0.1-1& $Days^{-1}$ & Assumed \\
$\delta$ & Apparition rate of vulnerability & 0.01-0.1& $Days^{-1}$ & Assumed\\
$\alpha_C$ & Production of \textit{V. cholerae} by asymptomatic& $10^4-10^5$& $cells.Days^{-1}$ & Assumed\\
$\alpha_I$ & Production of \textit{V. cholerae} by infected& $10^6-10^8$ & $cells.Days^{-1}$ & Assumed\\
$a$ & Nutrition rate & 0.5-1 & $Days^{-1}$ & Assumed\\
$e$ & Conversion rate of food consumption &$10^{-3}$ & $Unity/calories$ & Assumed\\
$\gamma_C$ & Recovered rate for asymptomatic individuals & 0.14-0.2 & $Days^{-1}$ & Assumed\\
$\gamma_I$ & Recovered rate for infected individuals & 0.1-0.14& $Days^{-1}$ & Assumed\\
$m_{ij}$ & Migration to patch i from patch j & 0.1-0.4& $Days^{-1}$ & Assumed\\
$\rho$ & Rate of \textit{V. cholerae} - foods association & 0.37& $Days^{-1}$ & Assumed\\
$\omega$ & Natural decay for contaminated foods & 0.15 &$Days^{-1}$ & Assumed\\
$\sigma$ & Colonization coefficient of bacteria & 0.3-0.5 & $Days^{-1}$ & Assumed\\
$p$ & Proportion of direct infectious state & 0.6 & $Days^{-1}$ & WHO, 2022\\
$\zeta$ & Rate to contract infection for individuals living in food security & 0.42 & $Days^{-1}$ & Assumed\\
$\Lambda_i$ & Recruitment rate at patch i& $10^3$ & $Days^{-1}$ & Assumed\\
$\pi$ & Proportion of recruitment for vulnerable & 0.3-0.5 & $Days^{-1}$ & Assumed\\
\hline
\end{tabular}
\label{par}
\end{center}
\end{sidewaystable}

\section{Mathematical analysis}\label{sec:Analysis}
\subsection{Basic properties}
First of all, we need to establish that system \eqref{systeme} is well-posed to ensure that the model makes sense biologically.

\begin{propo}\label{propo:positif}
The nonnegative orthant $\mathbb{R}_+^{8n}$ is positively invariant under the flow of system \eqref{systeme} if initial conditions satisfy $S_{1,i}(0), S_{2,i}(0), C_i(0), I_i(0), F_i^s(0), F_i^c(0), B_i(0) > 0$ for all $i$ and all other variables are nonnegative, then $X(t) \geq 0$ for all $t \geq 0$ and for all variables $X \in \{S_{1,i}, S_{2,i}, C_i, I_i, R_i, F_i^s, F_i^c, B_i\}$.
\end{propo}

\begin{proof}
Assume the initial conditions are nonnegative with $S_{1,i}(0), S_{2,i}(0), C_i(0), I_i(0), F_i^s(0),$ \\ $ F_i^c(0), B_i(0) > 0$ for all $i = 1,\dots,n$. Let $t_1 > 0$ be the minimal time at which some component $X_k(t_1) = 0$ for $X \in \{S_{1}, S_{2}, I, C, R, F^s, F^c, B\}$, and $X_k(t) > 0$ for all $t < t_1$. We investigate which component can be the first to exit the nonnegative orthant.

Suppose that for some $i$, $S_{1,i}(t_1) = 0$. From the corresponding equation,
\begin{equation}
\dot{S}_{1,i} = (1-\pi)(\Lambda_i + \xi R_i) + \phi_2S_{2,i} +  - \zeta\lambda(.)S_{1,i} - (\phi_1 + \mu) S_{1,i} + \sum_{j=1}^{n} (m_{ji}^S S_{1,j}-m_{ij}^S S_{1,i}) \, ,
\end{equation}

we observe that at $t_1$, all negative terms vanish due to $S_{1,i}(t_1) = 0$, and the remaining terms $\phi_2 S_{2,i}(t_1)$ and $\sum_{j} m_{ji}^S S_{1,j}(t_1)$ are nonnegative and strictly positive since $S_{2,i}(t_1) > 0$ and $S_{1,j}(t_1) > 0$ for $j \neq i$ (as $S_{1,i}$ is the first to reach zero). Therefore,
\[
\dot{S}_{1,i}(t_1) > 0,
\]
And it comes that
\[\int_{0}^{t_1} \dot{S}_{1,i}(t) dt >  0 \, ,\]
which implies  \begin{align*} 
&  S_{1,i}(t_1) - S_{1,i}(0) > 0 \Longrightarrow  S_{1,i}(t_1) > S_{1,i}(0) > 0 \; .
\end{align*}         
Finally, one obtains $S_{1,i}(t_1) >0$, that contradict the above hypothesis. The same reasoning applies for $S_{2,i}$ using its equation.

Now suppose $I_i(t_1) = 0$ for some $i$, with $I_i(t) > 0$ for $t < t_1$. Then from
\begin{equation}
\dot{I}_i = p \lambda (S_{2,i} + \zeta S_{1,i}) - (\gamma_I + d(.) + \mu) I_i + \sum_{j} (m_{ji}^I I_j - m_{ij}^I I_i),
\end{equation}
we have $\dot{I}_i(t_1) = p \lambda (S_{2,i} + \zeta S_{1,i}) + \sum_j m_{ji}^I I_j > 0$, since $S_{1,i}, S_{2,i}, I_j > 0$ for all $j$ and $t < t_1$. This again contradicts the assumption that $I_i(t_1) = 0$ and decreasing.

Similar reasoning applies to the asymptomatic class $C_i$, where the derivative at $t_1$ takes the form:
\begin{equation}
\dot{C}_i = (1-p)\lambda(S_{2,i} + \zeta S_{1,i}) + \sum_j m_{ji}^C C_j > 0.
\end{equation}

For the recovered class $R_i$, we have
\begin{equation}
\dot{R}_i = \gamma_C C_i + \gamma_I I_i - (\mu + \xi)R_i+ \sum_j (m_{ji}^R R_j - m_{ij}^R R_i).
\end{equation}
At $t_1$, if $R_i(t_1) = 0$, using similar process the remaining terms are nonnegative and strictly positive, leading to $\dot{R}_i(t_1) > 0$.

Next, consider the food compartments. If $F_i^s(t_1) = 0$, from
\begin{equation}
\dot{F}_i^s = r_i F_i\left(1 - \frac{F_i^s}{K_1}\right) - (a_i + \omega) F_i^s - \rho B_i F_i^s,
\end{equation}
we have $\dot{F}_i^s(t_1) = r_i F_i > 0$ , since $F_i^c < K_1$ and $F_i = F_i^s + F_i^c > 0$.

Similarly, if $F_i^c(t_1) = 0$,
\[
\dot{F}_i^c = \rho B_i F_i^s > 0,
\]
since $B_i, F_i^s > 0$.

For the bacterial concentration $B_i$, the equation
\begin{equation}
\dot{B}_i = \alpha_C C_i + \alpha_I I_i - \mu_B(T) B_i + \sum_j (m_{ji}^B B_j - m_{ij}^B B_i)
\end{equation}
gives $\dot{B}_i(t_1) > 0$ because $C_i, I_i > 0$ and $B_i = 0$ eliminates the negative term.

Hence, no variable can be the first to exit the nonnegative orthant. Therefore, all components remain nonnegative for all $t \geq 0$, and $S_{1,i}(t), S_{2,i}(t), F_i^s(t)$ remain strictly positive.
\end{proof}

\begin{lem}\label{lem:bornitude}
The total human population $T_H(t) = \sum_{i=1}^n N_i(t)$ and the total bacterial concentration $T_B(t) = \sum_{i=1}^n B_i(t)$ are bounded for all $t \geq 0$.
\end{lem}

\begin{proof}
Consider first the evolution of the total human population:
\begin{equation}
T_H(t) = \sum_{i=1}^n N_i(t) = \sum_{i=1}^n \left(S_{1,i}(t) + S_{2,i}(t) + C_i(t) + I_i(t) + R_i(t)\right).
\end{equation}
By summing the differential equations for $S_{1,i}, S_{2,i}, C_i, I_i$ and $R_i$ over all $i = 1, \dots, n$, we obtain:
\begin{equation}
\frac{d T_H}{dt} = \sum_{i=1}^n \left( \dot{S}_{1,i} + \dot{S}_{2,i} + \dot{C}_i + \dot{I}_i + \dot{R}_i \right).
\end{equation}

Observe that all inter-patch movement terms cancel out in the sum. For example, migration of $S_{1,i}$ from patch $i$ to patch $j$ appears with opposite signs in the equations for $S_{1,i}$ and $S_{1,j}$:
\begin{equation}
\sum_{i=1}^n \sum_{j=1}^n (m_{ji}^X X_j - m_{ij}^X X_i) = 0.
\end{equation}
Now, from the system \eqref{systeme}, 
suppose \[\mu^H = \min_{i=1,\ldots,n} \mu_i,\]

Thus,
\begin{equation}
\frac{d T_H}{dt} \leq \Lambda^* -\mu^H T_H(t).
\end{equation}

This implies that the total human population is non-increasing over time. And using Proposition \ref{propo:positif}, hence for all $t \geq 0$:
\begin{equation}
0 \leq T_H(t) \leq \max \left\{T_H(0), \frac{\Lambda^*}{\mu^H}\right\}.
\end{equation}
with \[\Lambda = \sum_{i=1}^{n} \Lambda_i. \]
Now consider the total bacterial concentration:
\begin{equation}
T_B(t) = \sum_{i=1}^n B_i(t).
\end{equation}
From the bacterial dynamics:
\begin{equation}
\dot{B}_i = \alpha_C C_i + \alpha_I I_i - \mu_{B_i}(T) B_i- \sigma\rho B_iF_{i}^s + \sum_{j=1}^{n} (m_{ji}^B B_j - m_{ij}^B B_i),
\end{equation}
the migration terms again cancel upon summing over $i$.

Thus:
\begin{equation}
\frac{d T_B}{dt} = \sum_{i=1}^n (\alpha_C C_i + \alpha_I I_i - \sigma\rho B_iF_{i}^s- \mu_{B_i}(T) B_i).
\end{equation}

Let $\alpha_{\max} = \max(\alpha_C, \alpha_I)$ (intuitively $\alpha_{\max} = \alpha_I$), $C_{\max}(t) = \max_i C_i(t)$, and similarly for $I_i$ (because $C_i, I_i$ are bounded). Then:
\begin{equation}
\frac{d T_B}{dt} \leq \alpha_{\max} \sum_{i=1}^n (C_i + I_i) - \mu_{\min} T_B,
\end{equation}
where $\mu_{\min} = \min_{i,t} \mu_{B_i}(T(t)) > 0$ (by hypothesis that $T(t) < T_{\max}$ and temperature is bounded in real situations).

Let $K = \alpha_{\max} \cdot \max_{t} T_H(t)$. Then:
\begin{equation}
\frac{d T_B}{dt} \leq K - \mu_{\min} T_B(t).
\end{equation}

This is a linear differential inequality of the form:
\begin{equation}
\frac{d T_B}{dt} + \mu_{\min} T_B \leq K.
\end{equation}
Integrating, we find:
\begin{equation}
T_B(t) \leq e^{-\mu_{\min} t} \left(T_B(0) - \frac{K}{\mu_{\min}} \right) + \frac{K}{\mu_{\min}}.
\end{equation}
Therefore, $T_B(t)$ is bounded for all $t \geq 0$, and in particular:
\begin{equation}
0 \leq T_B(t) \leq \max\left( T_B(0), \frac{K}{\mu_{\min}} \right).
\end{equation}

Combining both bounds, all human compartments and the bacteria concentration remain bounded over time. Since the other compartments are involved in equations with logistic or saturating growth, bounded by carrying capacities $K_1$ and food production rates $r_i$, they are also bounded.

Thus, all variables of system \eqref{systeme} are uniformly bounded, and the solution remains in a compact, positively invariant subset of $\mathbb{R}_+^{8n}$.
\end{proof}

\begin{ppte}\label{ppte:domaine}
Let the initial conditions satisfy $X(0) \in \Omega$, where
\begin{equation}
\Omega = \left\{ X \in \mathbb{R}_+^{8n} \;\middle|\;
\begin{array}{l}
0 \leq T_H(t) \leq M_H, \\
0 \leq T_B(t) \leq M_B, \\
0 \leq F_i^s(t) + F_i^c(t) \leq K_1,\quad \forall i = 1,\ldots,n
\end{array}
\right\}.
\end{equation}
with,
\[
M_H = \max \left\{T_H(0), \frac{\Lambda}{\mu^H}\right\} \quad \text{and} \quad M_B = \max\left( T_B(0), \frac{K}{\mu_{\min}} \right).
\]
Then, the region $\Omega$ is positively invariant under the flow of system \eqref{systeme}. That is, any solution $X(t)$ with $X(0) \in \Omega$ remains in $\Omega$ for all $t \geq 0$.
\end{ppte}

\begin{proof}
From Proposition \ref{propo:positif}, we know that all components of the solution remain nonnegative, i.e., $X(t) \in \mathbb{R}_+^{8n}$ for all $t \geq 0$. Furthermore, from Property \ref{ppte:domaine}, the total human population satisfies $T_H(t) = \sum_{i=1}^n N_i(t) \leq M_H$, and the total bacterial concentration satisfies $T_B(t) = \sum_{i=1}^n B_i(t) \leq M_B$ for all $t \geq 0$.

Moreover, since the food compartment $F_i(t) = F_i^s(t) + F_i^c(t)$ evolves according to a logistic growth law with upper bound $K_1$, it follows that $F_i(t) \leq K_1$ for each $i = 1, \ldots, n$.

Hence, all conditions defining the set $\Omega$ are preserved for all $t \geq 0$, which proves that $\Omega$ is positively invariant under the flow of system \eqref{systeme}.
\end{proof}

\begin{theo}[Global existence and boundedness of solutions]
Let the initial conditions $X(0) \in \Omega \subset \mathbb{R}_+^{8n}$, and suppose that the right-hand side of system \eqref{systeme} is locally Lipschitz-continuous in $X$. Then, the system \eqref{systeme} admits a unique global solution $X(t)$ for all $t \geq 0$, which remains in the positively invariant set $\Omega$.
\end{theo}

\begin{proof}
The system \eqref{systeme} is a system of ordinary differential equations with locally Lipschitz-continuous right-hand side in  $\mathbb{R}_+^{8n}$. Thus, by the Cauchy-Lipschitz theorem, a unique local solution $X(t)$ exists for some maximal time interval $[0, t_{\max})$.

From Proposition \ref{propo:positif}, the solution remains non-negative for all $t \in [0, t_{\max})$. Lemma \ref{lem:bornitude} ensures that the total human population and bacterial concentration remain bounded, and Property \ref{ppte:domaine} confirms that the compact positively invariant set $\Omega$ is preserved under the flow of the system.

Therefore, the solution remains bounded and cannot blow up in finite time, implying that $t_{\max} = +\infty$. Consequently, the local solution extends to a global solution. Since the solution stays in $\Omega$ for all $t \geq 0$, the proof is complete.
\end{proof}

\begin{remark}
\begin{itemize} \item[\color{white}]
\item Movements within the same patch are ignored ($m_{ii} = 0, \;\forall i$).
\item In this paper, authors assimilate vulnerable individuals at individuals in food insecurity.
\item Generally, when we have patch model in case without movement ($m_{ji}^S=...=m_{ji}^B=0$), using \cite{arinocity} the basic reproduction number may be given by:
\begin{equation}
\mathcal{R}_0^{\text{global}} = \max_{i=1,2,3,4} \mathcal{R}_{0,i}.
\end{equation}
\end{itemize}
\end{remark}
\subsection{Disease-Free Equilibrium}

At the disease-free equilibrium (DFE), the reduced system becomes:

\[
\left\{
\begin{array}{cl}
\dot{S}_{1,i} &= (1-\pi)\Lambda + \xi_1 R_i + \phi_2S_{2,i} - (\phi_1 + \mu) S_{1,i} + \sum_{j=1}^{n} (m_{ji}^S S_{1,j}-m_{ij}^S S_{1,i}), \\
\dot{S}_{2,i} &= \pi\Lambda + \xi_2 R_i + \phi_1 S_{1,i} - (\phi_2 + \mu) S_{2,i} + \sum_j (m_{ji}^S S_{2,j} - m_{ij}^S S_{2,i}), \\
\dot{R}_i &= - (\mu + \xi_1 + \xi_2)R_i + \sum_j (m_{ji}^R R_j - m_{ij}^R R_i), \\
\dot{F}_i^s &= r_i F_i^s \left(1 - \frac{F_i^s}{K_1}\right) - (a_i + \omega) F_i^s.
\end{array}
\right.
\]

Let us define the following vectors:

\[
S_1 = (S_{1,1}, \ldots, S_{1,n})^\top,\quad
S_2 = (S_{2,1}, \ldots, S_{2,n})^\top,\quad
R = (R_1, \ldots, R_n)^\top
\]

Let $\mathcal{M}^S$ and $\mathcal{M}^R$ be the movement matrices for susceptible and recovered individuals, respectively, and consider all the other parameters in matrix form such as

\[
\phi_1 = \mathrm{diag}(\phi_1^{(i)}),\quad \phi_2 = \mathrm{diag}(\phi_2^{(i)})
\]
\begin{ppte}
$(\phi_1 + \phi_2 + \mu^H - \mathcal{M}^S)^{-1}$ is nonsingular and $(\phi_1 + \phi_2 + \mu^H - \mathcal{M}^S)^{-1}>>0$.
\end{ppte}
\begin{proof}
The proof is guaranteed by \cite{Arino2019}, Proposition 3.
\end{proof}
The linearized system at the DFE becomes:

\[
\left\{
\begin{array}{l}
\cfrac{d}{dt} \begin{pmatrix} S_1 \\ S_2 \end{pmatrix}
= \begin{pmatrix}
(1-\pi)\Lambda \\
\pi\Lambda 
\end{pmatrix} +
\begin{pmatrix}
\mathcal{M}^S - (\phi_1 + \mu^H) & \phi_2 \\
\phi_1 & -(\phi_2 + \mu^H)+ \mathcal{M}^S
\end{pmatrix}
\begin{pmatrix} S_1 \\ S_2 \end{pmatrix}
= 0 \\
\cfrac{dR}{dt} =  -\mu R +\mathcal{M}^R R \Rightarrow R^* \in \ker(\mathcal{M}^R) \quad \text{(since $\mathcal{M}^R$ is irreducible)}
\end{array}
\right.
\]

Then, define:

\[
A = \begin{pmatrix}
\mathcal{M}^S - (\phi_1 + \mu^H) & \phi_2 \\
\phi_1 & -(\phi_2 + \mu^H)+ \mathcal{M}^S
\end{pmatrix}, \quad \text{and} \quad b = \begin{pmatrix} b_1 \\ b_2 \end{pmatrix}.
\]
with: $$b_1 = \Lambda(1-\pi), b_2 = \pi\Lambda. $$

One has
\begin{equation*}
S_1^* =  ((\mathcal{M}^S - \mu^H)(\mathcal{M}^S-\mu^H-\phi_1-\phi_2))^{-1}((1-\pi)(-\mathcal{M}^S +\mu^H)+\phi_2)\Lambda,
\end{equation*}
\begin{equation}
S_2^* =  ((\mathcal{M}^S - \mu^H)(\mathcal{M}^S-\mu^H-\phi_1-\phi_2))^{-1}(\pi(\mu^H-\mathcal{M}^S)+\phi_1)\Lambda.
\end{equation}

Then, the equilibrium is given by:
\begin{equation}\label{eq:equilibre}
\mathcal{E}_0^{(*)} = \left(S^*_1, S^*_2, \ldots, F^*_s, 0, 0\right),
\end{equation}
This expression defines the disease-free equilibrium distribution of susceptible individuals, accounting for nutritional transitions and mobility.

\begin{con}
Using the above result in Eq.\ref{eq:equilibre} without migration there exists two (02) diseases-free equilibria $\mathcal{E}_0^{(1)}, \mathcal{E}_0^{(2)} \in \mathbb{R}_+^{8n}$ for system \eqref{systeme}.
\begin{enumerate}
\item A state of equilibrium with susceptible individuals but without food.
\[
\mathcal{E}_0^{(1)} = \left(S^0_1, S^0_2, \ldots, 0\right),
\]
\item A state of equilibrium where there are individuals and food.
\[
\mathcal{E}_0^{(2)} = (S^0_1, S^0_2, 0, 0, 0, F_s^0, 0, 0).
\]
\end{enumerate}
where $S^0_1 = \cfrac{\Lambda(\phi_2 + \mu(1-\pi))}{\mu(\mu+\phi_1 + \phi_2)}, \; S^0_2 = \cfrac{\Lambda(\phi_1 + \mu\pi)}{\mu(\mu +\phi_1 + \phi_2)}$ and $F_s^0 = K\left(1-\cfrac{a+\omega}{r_i}\right)$.
\end{con}

 \subsection{Basic reproduction number}
 Let us consider the system \eqref{systeme} without inter-patch movements (i.e $m_{ji}^X = 0$, for all variables $X$ and all $i,j$). We find disease-free equilibria, and they occur when : $I = C = B = F^c = 0$.
 
 To compute this basic reproduction number, we use the method of Van den \cite{Driessche2002}. To do so, system \eqref{systeme} can be written in the following form:
 
 \begin{equation}
 \dot{X}=\mathcal{F}(X)-\mathcal{V}(X),
 \end{equation}
 where $X=(C,I,F^s,B)$,
 
Then inside a satellite, we have:
  
  \begin{align}
  \mathcal{F}(X) & = \begin{pmatrix}
    (1-p)\left(\beta_H\cdot (I + \varepsilon C){N}\right)(S_2 + \zeta S_1)\\
    \\
    p\left(\beta_H\cdot (I_i + \varepsilon C)\right)(S_2 + \zeta S_1)\\ 0 \\ 0
   \end{pmatrix}\,\,\ \mbox{and}&\mathcal{V}(X) & =  \begin{pmatrix}
     (\gamma_C+\mu) C\\
     \\
    (\gamma_I + d+\mu)I\\ -\rho BF^s + (a+\omega)F^c\\
    -(\alpha_CC + \alpha_II) + \mu_B
    \end{pmatrix}.  
  \end{align}
 
 We point that $\mathcal{F}(X)$ is the speed of appearance of newly infected humans. These are the newly infected obtained by transmissions of all kinds. and $\mathcal{V}(X)$ is the rate of onset of new cases for reasons other than disease.
  
  The Jacobian matrix of $\mathcal{F}(X)$ and $\mathcal{V}(X)$ at the disease free equilibrium point $\mathcal{E}_0^{(4)}$ are:
  
\begin{equation}
\frac{d\mathcal{F}}{dy} = 
\begin{pmatrix}
  (1-p)\beta_H\varepsilon(S_2^* + \zeta S_1^*) & (1-p)\beta_H(S_2^* + \zeta S_1^*) & (1-p)\beta_B(S_2^* + \zeta S_1^*) & 0\\
  \\
 p\beta_H\varepsilon(S_2^* + \zeta S_1^*) &  p\beta_H(S_2^* + \zeta S_1^*)&p\beta_B(S_2^* + \zeta S_1^*) & 0\\
 0 & 0 & 0 & 0\\
 0 & 0 & 0 & 0
   \end{pmatrix},
\end{equation}
and
\begin{equation}
{\frac{d\mathcal{V}}{dy}\lvert_{X=\mathcal{E}_0^{(2)}}}^{-1} =  \left(\begin{matrix}\cfrac{1}{\gamma_{C}+\mu} & 0 & 0 & 0\\0 & \cfrac{1}{d + \gamma_{I}+\mu} & 0 & 0\\\cfrac{F_s^0 \alpha_{C} \rho}{\gamma_{C} \mu_{B} \left(a + \omega\right)} & \cfrac{F_s^0 \alpha_{I} \rho}{\mu_{B} \left(a d + a \gamma_{I} + d \omega + \gamma_{I} \omega\right)} & \cfrac{1}{a + \omega} & \cfrac{F_s^0 \rho}{\mu_{B} \left(a + \omega\right)}\\\cfrac{\alpha_{C}}{\gamma_{C} \mu_{B}} & \cfrac{\alpha_{I}}{\mu_{B} \left(d + \gamma_{I}\right)} & 0 & \cfrac{1}{\mu_{B}}\end{matrix}\right),
\end{equation}
 Then,
{\scriptsize \begin{equation}
\mathcal{F}\mathcal{V}^{-1} = j
\displaystyle \left(\begin{matrix}\cfrac{m \left(F_s^0 \alpha_{C} \beta_{B} \rho + \beta_{H} \epsilon \mu_{B} \left(a + \omega\right)\right)}{\gamma_{C} \mu_{B} \left(a + \omega\right)} & \cfrac{m \left(F_s^0 \alpha_{I} \beta_{B} \rho + a \beta_{H} \mu_{B} + \beta_{H} \mu_{B} \omega\right)}{\mu_{B} \left(a d + a \gamma_{I} + d \omega + \gamma_{I} \omega\right)} & \cfrac{\beta_{B} m}{a + \omega} & \cfrac{F_s^0 \beta_{B} m \rho}{\mu_{B} \left(a + \omega\right)}\\\cfrac{p \left(F_s^0 \alpha_{C} \beta_{B} \rho + \beta_{H} \epsilon \mu_{B} \left(a + \omega\right)\right)}{\gamma_{C} \mu_{B} \left(a + \omega\right)} & \cfrac{p \left(F_s^0 \alpha_{I} \beta_{B} \rho + a \beta_{H} \mu_{B} + \beta_{H} \mu_{B} \omega\right)}{\mu_{B} \left(a d + a \gamma_{I} + d \omega + \gamma_{I} \omega\right)} & \cfrac{\beta_{B} p}{a + \omega} & \cfrac{F_s^0 \beta_{B} p \rho}{\mu_{B} \left(a + \omega\right)}\\0 & 0 & 0 & 0\\0 & 0 & 0 & 0\end{matrix}\right).
 \end{equation}}
with
\[j = S_2^* + \zeta S_1^*, \quad \text{and} \quad m = 1- p.\]
One has,
\[
Sp(\mathcal{F}\mathcal{V}^{-1}) = \left\{0, \mathcal{R}_0\right\}.
\]
Using Theorem 2, \cite{Driessche2002} \[\mathcal{R}_0 = \rho(\mathcal{F}\mathcal{V}^{-1})\]

Yields
\begin{equation}
\mathcal{R}_0 = (S_2^0 + \zeta S_1^0)(\mathcal{R}_0^H  + \mathcal{R}_0^E).
\end{equation}
\begin{equation}
\mathcal{R}_0^H = \beta_H\left(\frac{(1-p)\varepsilon}{\gamma_C + \mu} + \frac{p}{\gamma_I + d(.) + \mu}\right) \quad \text{and} \quad \mathcal{R}_0^E = \beta_B\frac{F_s^0\rho}{\mu_B(a+\omega)}\left(\frac{(1-p)\alpha_C}{\gamma_C}  + \frac{p\alpha_I}{\gamma_I + d}\right).
\end{equation}
Where  $\mathcal{R}_0^H$ represents the transmission of the disease via human interactions.

Therefore, if a patch $i$, is isolated from the others i.e $m_{ji}^X = 0$, for all variables $X$ and all $i,j$, then the basic reproduction number in each patch is given by:
\begin{equation}
\mathcal{R}_0 =  \left\{ \begin{array}{ll} 
 \mathcal{R}_0^{(1)} = \cfrac{S_2^0 + \zeta S_1^0}{N_c}(\mathcal{R}_0^H  + \frac{\mathcal{R}_0^E}{\chi}) & \mbox{, \text{for the principal node}}\\ 
 \mathcal{R}_0^{(i)} =  (S_2^0 + \zeta S_1^0)(\mathcal{R}_0^H  + \mathcal{R}_0^E) & \mbox{, \text{for their satellites}} \end{array} 
 \right. \end{equation}
with $N_c$ the total population inside the central node.

\begin{theo}
Suppose that $m_{ij}=m_{ji}=m$ for all i,j=1,...,n. Then
\begin{equation}
\min_{i=1,\ldots, n} \mathcal{R}_0^{i} \le \mathcal{R}_0 \le \max_{i=1,\ldots, n} \mathcal{R}_0^{i}
\end{equation}
\end{theo}
\begin{proof}
See \cite{arinocity}.
\end{proof}
\subsection{Sensitivity analysis}
\begin{figure}[H]
\centering
\includegraphics[width=0.49\textwidth]{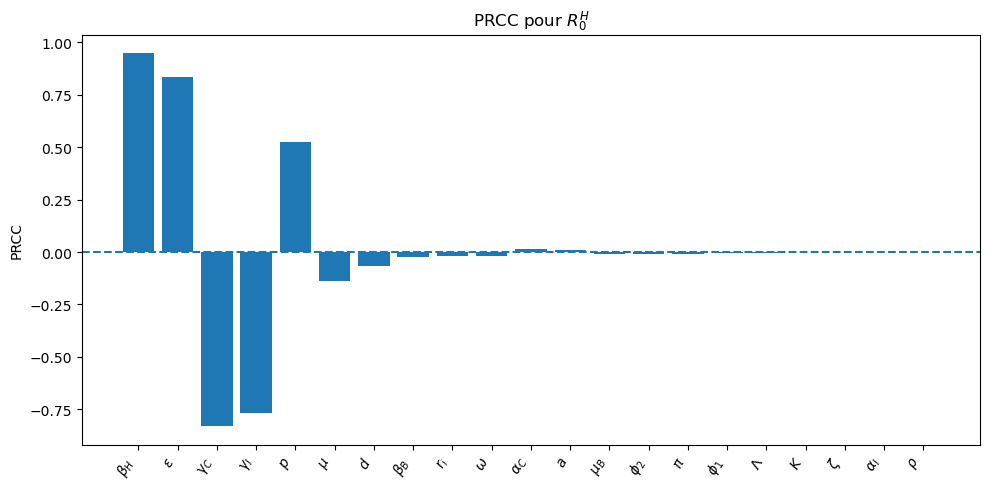}\hfill
\includegraphics[width=0.49\textwidth]{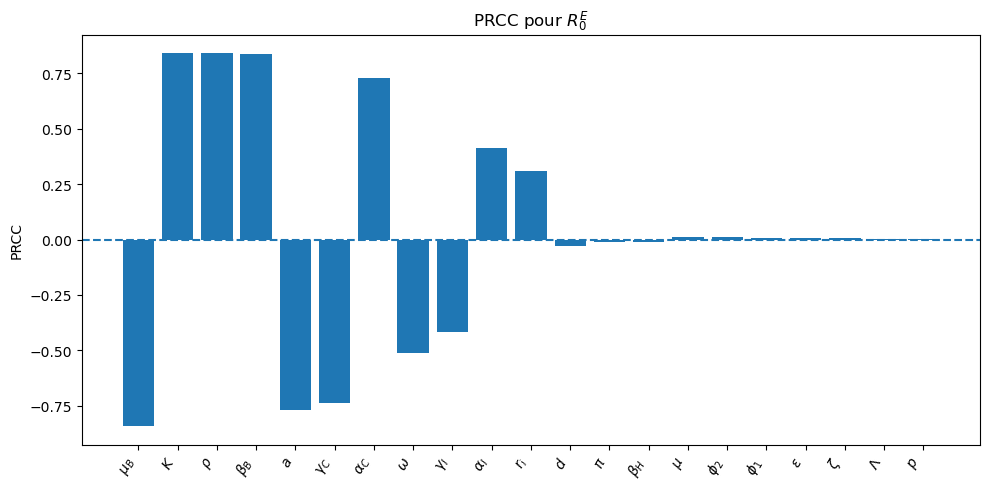}
\includegraphics[width=0.7\textwidth]{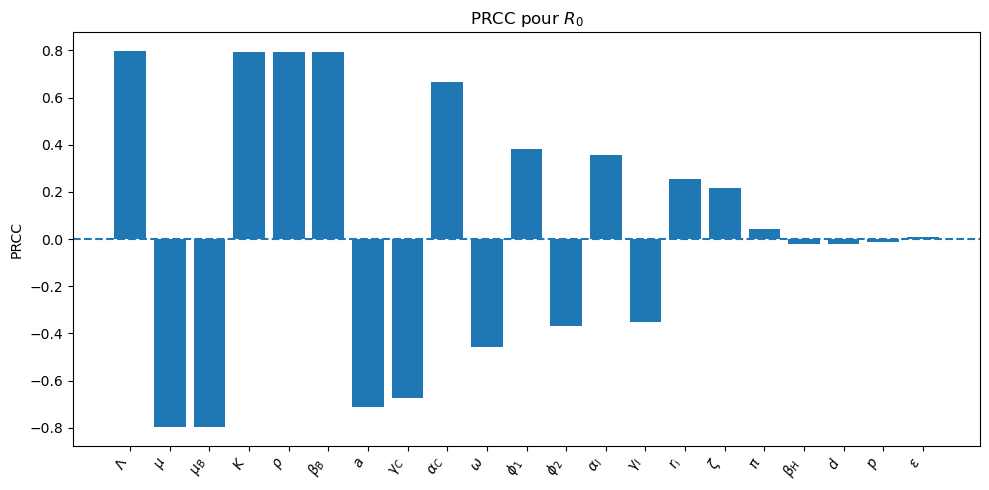}
\caption{Sensitivity of the number of cholera cases to changes in parameters in Tables 4 and 7 as computed by the Latin Hypercube Sampling-Partial Rank Correlation Coefficient (LHS-PRCC) index.}
\end{figure}
Figures $\mathcal{R}_0$, $\mathcal{R}_0^E$, and $\mathcal{R}_0^H$ display the results of a global sensitivity analysis of the basic reproduction number and its components using the LHS-PRCC method with uniform  parameter ranges. The overall $\mathcal{R}_0$ is most sensitive to environmental parameters, with $K$, $\rho$, and $\beta_B$  showing strong positive correlations and $\mu_B$, $a$, and $\omega$ showing strong negative effects, highlighting the central role of the environmental reservoir. For $\mathcal{R}_0^E$, sensitivity is dominated by bacterial shedding and clearance processes, whereas for $\mathcal{R}_0^H$, contact rate $\beta_H$, efficiency $\varepsilon$, and the proportion symptomatic $p$ are positively correlated, while recovery rates $\gamma_C$ and $\gamma_I$ have strong negative influence. These results emphasize that both sanitation/food safety and clinical interventions are essential to reducing cholera transmission.

\section{Model Application}\label{sec:Simulation}
This section presents numerical simulations that reflect local transmission patterns, spatial connectivity, and food-related vulnerability.

\subsection{Numerical simulations}
The initial conditions reflect a scenario where most individuals are susceptible, with a few infectious cases in the central node and no contamination in peripheral patches.
However, to assess the importance of dynamically modeling nutritional vulnerability, we simulated the model under varying transition rate $\phi_1$ between well-nourished and vulnerable individuals. Figure \ref{fig1} shows that epidemic size and mortality are highly sensitive to these parameters. In particular, increasing the rate $\phi_1$ at which individuals become vulnerable significantly amplifies both infection and death peaks. This justifies the inclusion of two susceptible compartments ($S_1$, $S_2$), as a single-compartment model with a static risk modifier would fail to capture these emergent dynamics.
\begin{figure}[H]
\centering
\includegraphics[width=0.49\textwidth]{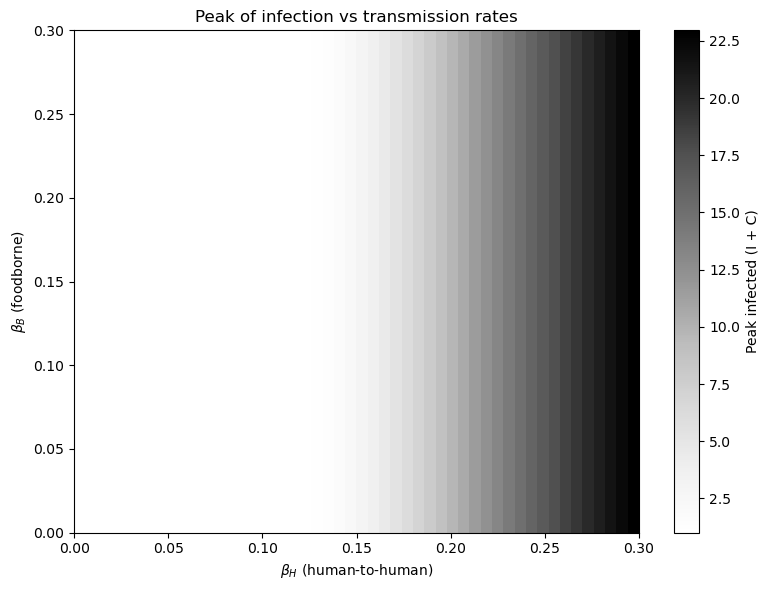}
\includegraphics[width=0.49\textwidth]{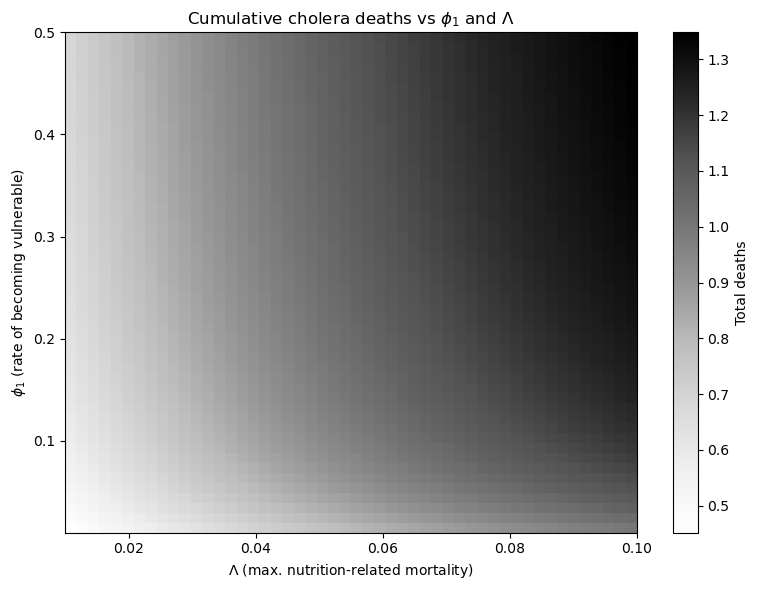}
\label{fig1}
\caption{Impact of mode of transmission inside infection transmission at right and effects of $\phi_1$ and $\Lambda$ on deaths dynamics at left.}
\end{figure}
Moreover, it confirms that nutritional transitions drive epidemic severity and that one class susceptible model would not capture these nonlinear effects.

Figure \ref{fig2} illustrates periodic dynamics in the population, which may indicate the presence of a limit cycle in the model. But, after a better observation, we notice a pseudo periodicity as well as a variation of the amplitudes of the variables which reflects other dynamic properties of system \eqref{systeme} probably due to the mobility of the infected which causes an emergence of the disease in healthy areas which is better observe in Fig \ref{fig:simulation}.
\begin{figure}[H]
\centering
\includegraphics[width=0.49\textwidth]{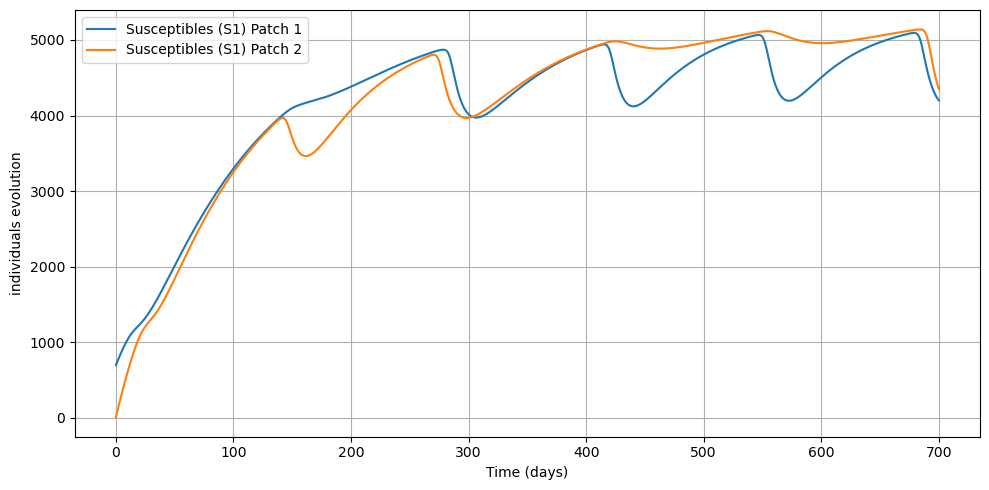}
\includegraphics[width=0.49\textwidth]{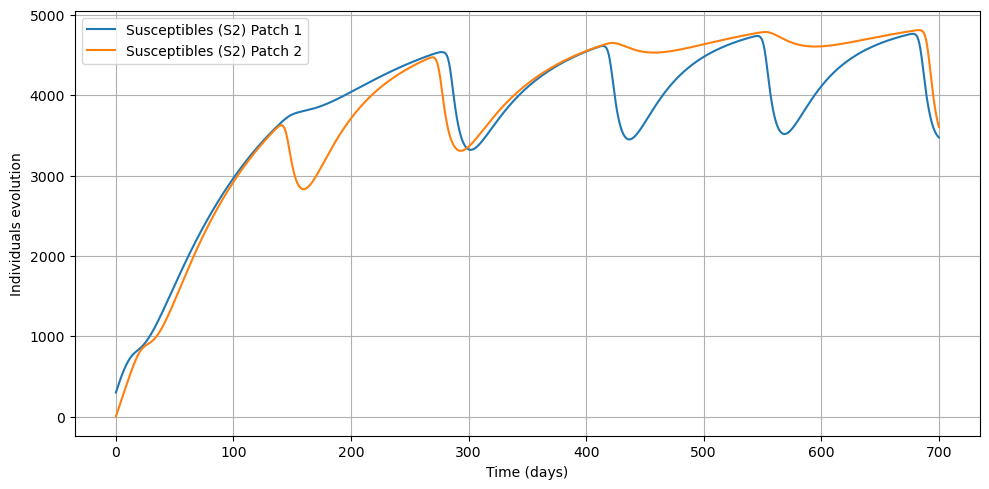}
\includegraphics[width=0.49\textwidth]{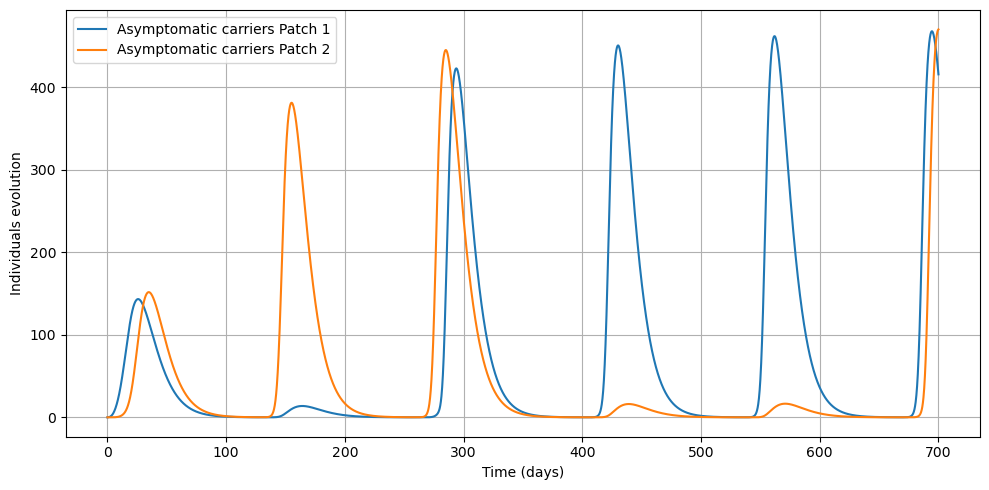}
\includegraphics[width=0.49\textwidth]{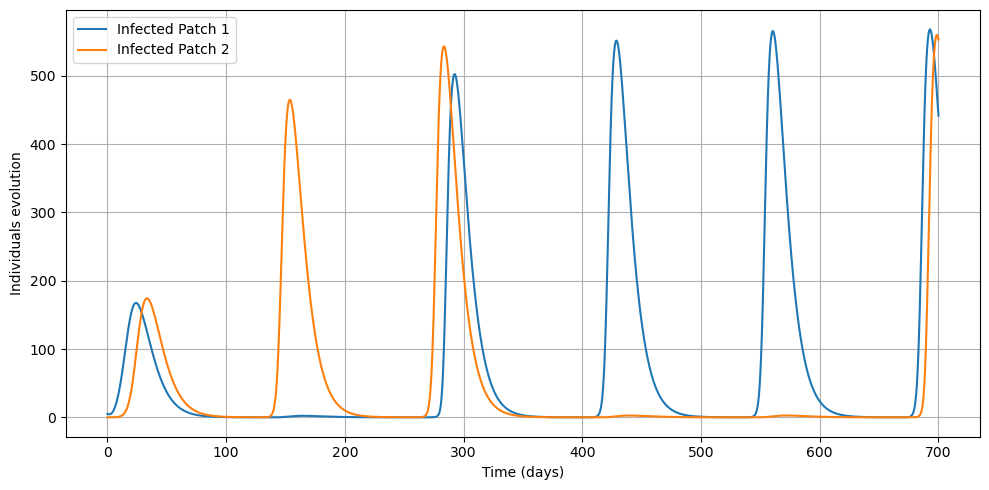}
\includegraphics[width=0.49\textwidth]{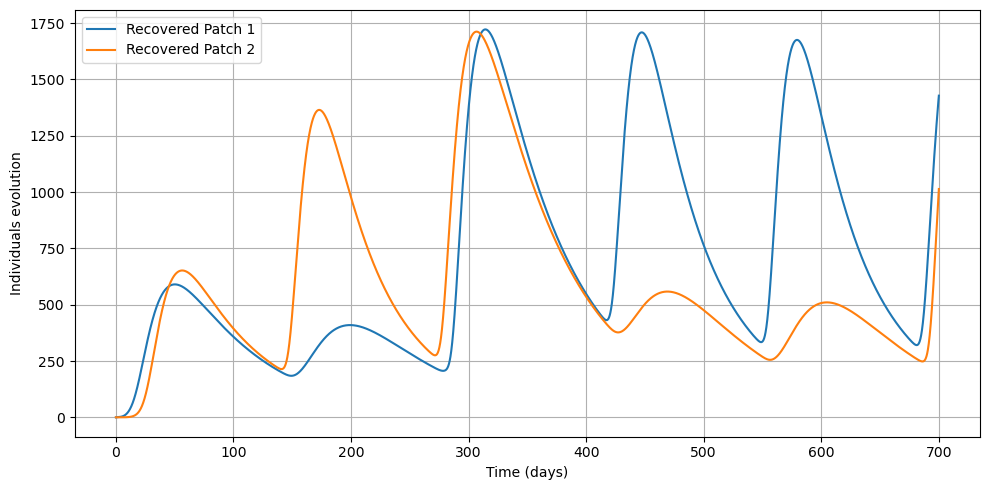}
\includegraphics[width=0.49\textwidth]{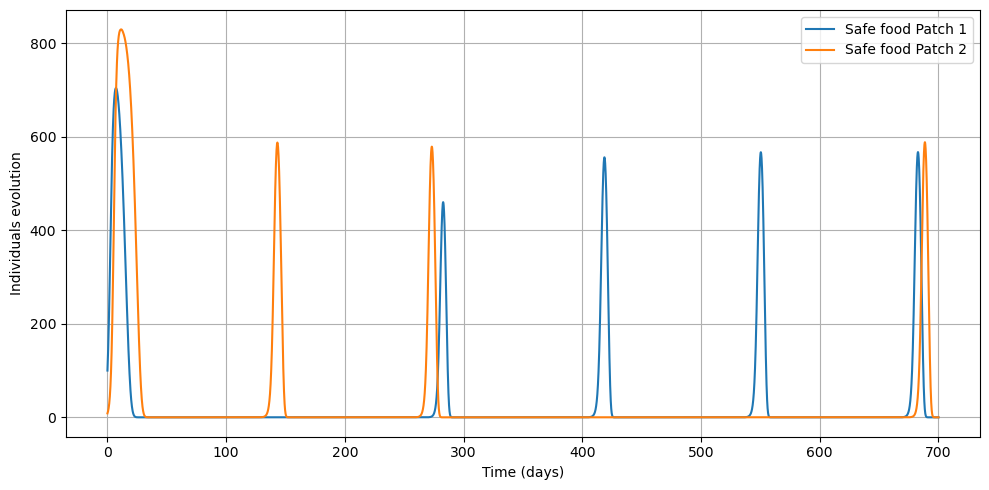}
\includegraphics[width=0.49\textwidth]{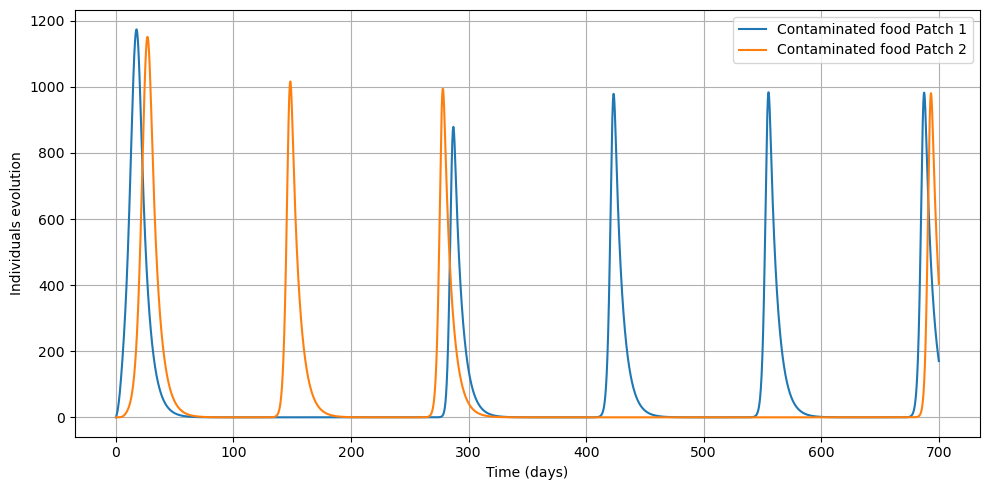}
\includegraphics[width=0.49\textwidth]{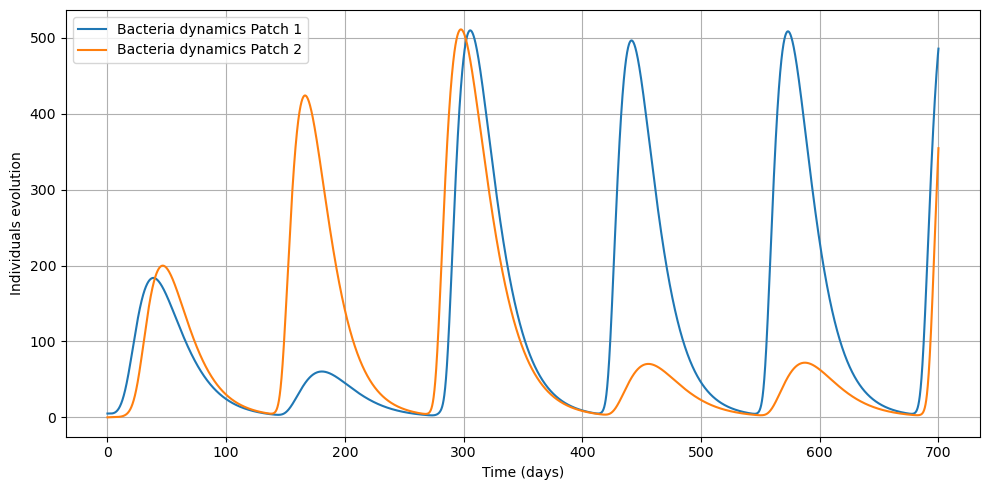}
\caption{Dynamics of model \eqref{systeme} with initial conditions: $S_1(0) = 700, S_2(0)=300, C(0) = 0, I(0)=5, R(0)=0, F^s(0)=100, F^c(0) = 0$ and $B(0)=0$ with $\mathcal{R}_0 = 31.7508$.}
\label{fig2}
\end{figure}

\subsection{Case study of Littoral (Douala) and its Surroundings Areas}

The modeling processes start inside two regions (SU, SW) and finish at  (Littoral and South-west regions) cover an area of more than 100000 km$^2$.
However, we focus the study inside Littoral region (particularly on the city of Douala) and its neighbors, which recorded one of the highest cholera case-fatality rates in Cameroon during the 2022 outbreak ($\approx 3\%$), and lies adjacent to the South-West region, the most affected area that epidemic. In a context marked by sociopolitical instability and high population displacement from conflict-affected \textquotedblleft anglophone\textquotedblright \, regions, Littoral serves as the main entry point for migrants from the South-West into \textquotedblleft francophone\textquotedblright \, urban zones. This demographic pressure, combined with preexisting challenges related to water, sanitation, and food safety, creates conditions conducive to cholera transmission and persistence.

To capture both local dynamics and spatial interactions, we consider a multi-patch framework composed of four interconnected zones:
\begin{itemize}
\item a central node representing Douala,
\item and three surrounding patches: Bonaberi, Bomono, and Yassa connected to Douala through human movement and food exchange.
 \item $\mu_B(T) = \tilde{\mu_B}\left(1-\kappa\cfrac{T - \bar{T}}{T_{\max} - \bar{T}}\right)$, where the temperature is in degree Celsius. $\kappa$ represents the dependency on temperature, ($\kappa = 3$ according to \cite{Bertuzzo2010}), $T_{\max}$ and $\bar{T}$ correspond respectively to the maximum and mean temperature of Douala city over the 20 years.
\end{itemize}

 This structure reflects the urban-peripheral gradient and allows us to account for spatial heterogeneity in exposure, vulnerability, and mobility, essential to understand the spread and control of cholera in the Douala metropolitan area.
  \begin{figure}[H]
  \centering
  \includegraphics[scale=0.45]{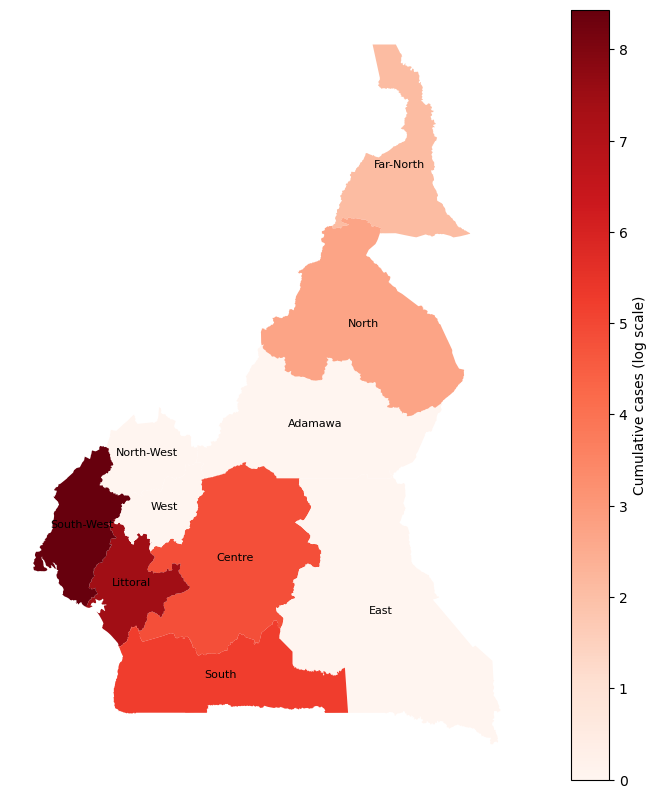}
  \includegraphics[scale=0.45]{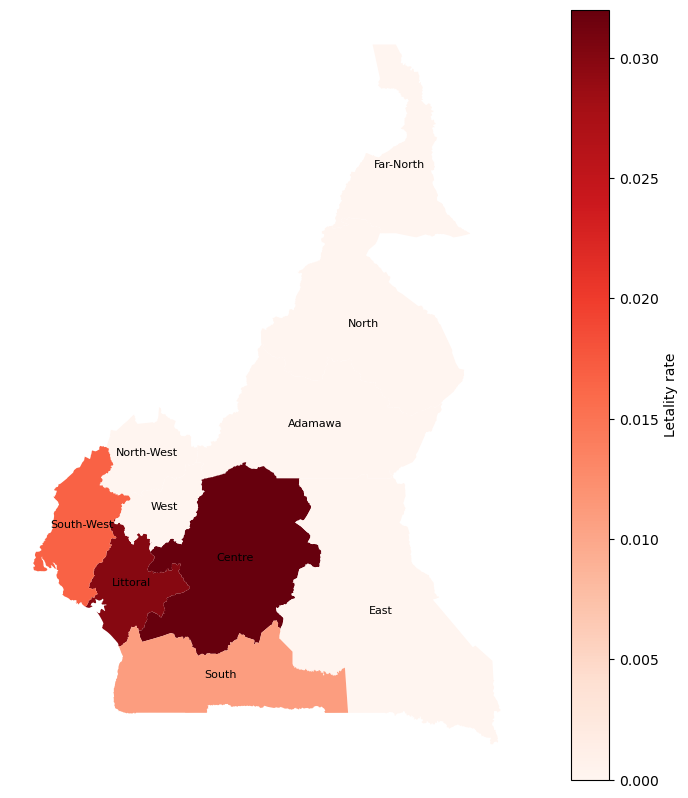} 
  \caption{Cumulative cholera cases repartition at the end of the 2022's epidemic in Cameroon according to WHO Database at left and letality rate due to the disease at right.}
   \label{fig:real-data}
  \end{figure}
 The left panel shows the cumulative distribution of reported cholera cases in Cameroon at the end of the 2022 epidemic, based on WHO data. The highest case counts are concentrated in the South-West and Littoral regions, reflecting intense transmission in these areas during the outbreak. The right panel depicts the corresponding case fatality rates (CFR). The Centre region exhibits the highest CFR, followed by the Littoral region. According to WHO reports, the outbreak began in the Centre region, which may explain the elevated CFR there early in the epidemic, limited preparedness and delayed response could have contributed to higher mortality. The spatial mismatch between incidence and CFR underscores that regions with fewer cases can still experience severe outcomes if access to timely treatment is inadequate. 

  \begin{figure}[H]
  \centering
  \includegraphics[scale=0.46]{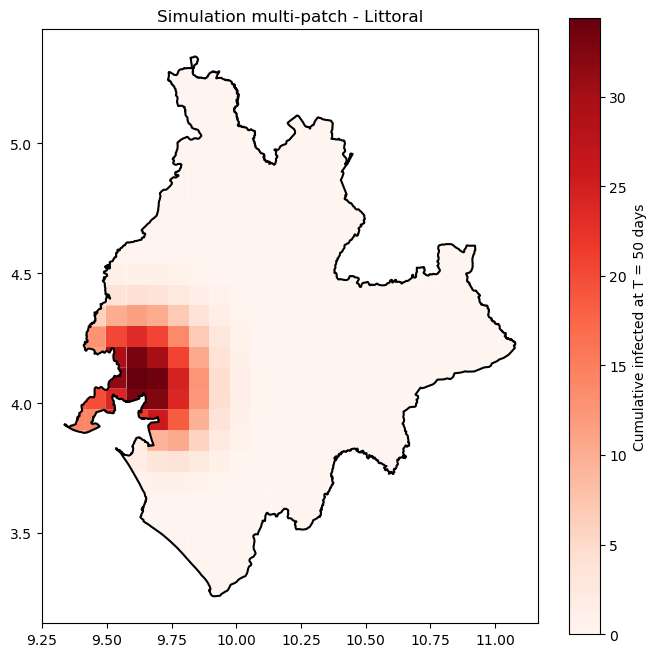}
  \includegraphics[scale=0.46]{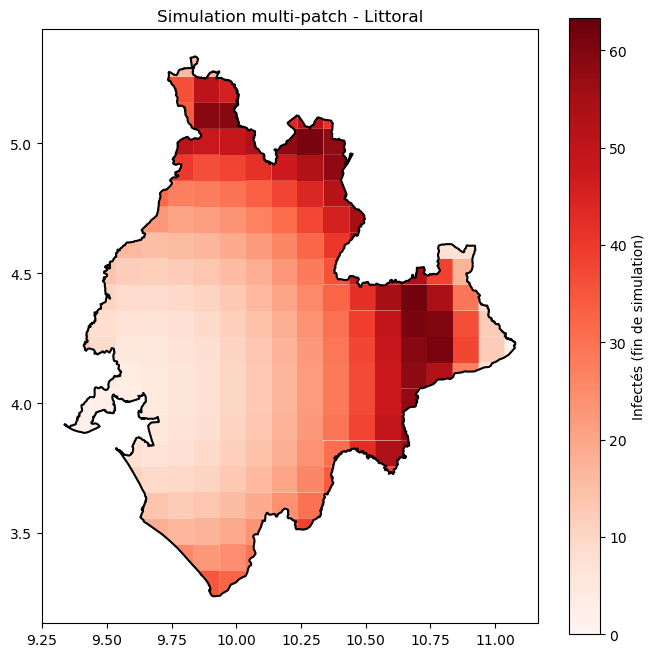}
  \caption{Propagation of the disease in Littoral region base on SCIRW-F \eqref{systeme} simulations during one year.}
  \label{fig:simulation}
  \end{figure}
 The simulated spatial spread (Figure \ref{fig:simulation}) of cholera in the Littoral region over one year, based on the SCIRW-F multi-patch system \eqref{systeme}, is shown for two time points. At left we represent the early phase of the epidemic, when infection is initially concentrated in central urban areas (notably Douala), before gradually diffusing to peripheral patches. The right panel shows the later stage of the outbreak, where high prevalence is observed in both central and peripheral zones due to human mobility and contaminated food exchange. This simulated pattern reproduces the observed urban to peripheral gradient seen in Figure \ref{fig:real-data} and emphasizes the role of spatial connectivity and vulnerability in sustaining transmission.

\section{Discussions and Perspectives}\label{sec:perspectives}
In this study, we developed a novel mathematical model to understand cholera transmission in the context of environmental contamination and food insecurity. The model progressively incorporates the key pathways of cholera spread from basic SIR dynamics, to waterborne transmission (SIWR), and finally to food/biomass mediated transmission via the proposed SIWR-F \eqref{sys} framework.

Moreover, the multi-patch extension \eqref{systeme} captures the geographic and demographic heterogeneity of cholera dynamics, emphasizing the role of migration, patch connectivity, and spatially varying food availability. Simulations based on the Douala metropolitan area illustrate how vulnerable populations in connected patches can act as persistent sources or sinks for the disease.

In conclusion, this work provides a comprehensive framework to analyze cholera outbreaks in resource-limited settings, showing that food insecurity is not just a background condition but a core driver of epidemic dynamics. The model can inform targeted interventions that integrate food, water, and mobility data paying the way toward more resilient health systems in vulnerable regions.

As possible extensions of this work, one can assume without loss of generality to extend this model \eqref{systeme} over several years (given that cholera is a disease in several countries of sub-Saharan Africa) by change with parameters time dependent and also explore an optimal control.

\section{Appendix: Mathematical proofs and tools}
\paragraph*{Appendix A}\textbf{Proof of Theorem} \ref{th:persistence}\\
\medskip\noindent Step 1: linearization on the disease-free boundary and the invasion criterion.\\
Consider the infectious subsystem \((i,w)\) when \(s\) and \(f\) are frozen at the boundary values \(s^*=1\) and \(f^*\). The linearization at the DFE \((s,i,w,r,f)=(1,0,0,0,f^*)\) restricted to \((i,w)\) is
\[
\begin{pmatrix}\dot i\\\dot w\end{pmatrix}
= A \begin{pmatrix} i\\ w\end{pmatrix},\qquad
A:=\begin{pmatrix}
\beta_I-(\gamma+\mu) & \beta_W(1+f^*)\\[4pt]
\xi & -\xi
\end{pmatrix}.
\]
Let \(p(\lambda)=\det(A-\lambda I)\) be the characteristic polynomial. Evaluating at \(\lambda=0\) we get
\[
p(0)=\det\begin{pmatrix}\beta_I-(\gamma+\mu) & \beta_W(1+f^*)\\ \xi & -\xi\end{pmatrix}
= -\xi\big(\beta_I+\beta_W(1+f^*)-(\gamma+\mu)\big).
\]
Hence \(p(0)<0\) if and only if \(\beta_I+\beta_W(1+f^*)-(\gamma+\mu)>0\), i.e. \(R_0>1\).
But \(p(\lambda)\to+\infty\) as \(\lambda\to+\infty\), so if \(p(0)<0\) there exists a positive real root \(\lambda_0>0\) of \(p(\lambda)\). Thus \(A\) has a positive eigenvalue \(\lambda_0>0\). The corresponding eigenvector can be chosen strictly positive because off-diagonal entries of \(A\) are nonnegative and \(A\) is cooperative on \(\mathbb{R}^2_+\).

Therefore \(\mathcal{R}_0>1\) is equivalent to the linearized infectious subsystem at the DFE having exponential growth (positive principal eigenvalue), i.e. the DFE is linearly unstable w.r.t. infectious perturbations.

\medskip\noindent\textbf{Step 2: local exponential growth from small infections (linear comparison).}
Because \(K\) is compact and the vector field is continuous, there exists a neighborhood \(U\subset K\) of the DFE boundary point \((1,0,0,0,f^*)\) such that, for all \((s,f)\) with \(|s-1|<\delta\), \(|f-f^*|<\delta\) (for some small \(\delta>0\)), the Jacobian matrix of the \((i,w)\)-subsystem
\[
J(s,f):=\begin{pmatrix}
\beta_I s-(\gamma+\mu) & \beta_W(1+f)s\\[4pt]
\xi & -\xi
\end{pmatrix}
\]
satisfies
\[
J(s,f)\ge A - \eta I
\]
(componentwise) for some small \(\eta\in(0,\lambda_0/2)\). (This follows from continuity of the entries of \(J\) in \((s,f)\).) The inequality is understood entrywise; the system for \((i,w)\) is cooperative (all off-diagonal infection/transmission terms are nonnegative), hence one may apply the comparison principle for cooperative linear systems.

Fix \(\eta\in(0,\lambda_0/2)\). By continuity there exists \(\delta>0\) so that whenever \(|s(t)-1|<\delta\) and \(|f(t)-f^*|<\delta\) the inequality above holds. Define the exit time
\[
\tau:=\inf\{t\ge 0:\ |s(t)-1|\ge\delta\ \text{ or }\ |f(t)-f^*|\ge\delta\}.
\]
On the time interval \([0,\tau)\) the vector \((i(t),w(t))^\top\) satisfies the differential inequality (entrywise)
\[
\frac{d}{dt}\begin{pmatrix}i\\w\end{pmatrix} \ge (A-\eta I)\begin{pmatrix}i\\w\end{pmatrix}.
\]
By standard linear comparison, we then have for \(t\in[0,\tau)\)
\[
\begin{pmatrix}i(t)\\w(t)\end{pmatrix} \ge e^{(A-\eta I)t}\begin{pmatrix}i(0)\\w(0)\end{pmatrix}.
\]
Because \(A\) has principal eigenvalue \(\lambda_0>\,2\eta\), the matrix \(A-\eta I\) has a positive growth exponent \(\lambda_0-\eta>0\). Hence there exists \(T>0\) (independent of sufficiently small initial infectious perturbations) and a positive vector \(\phi\gg 0\) such that
\[
\begin{pmatrix}i(T)\\w(T)\end{pmatrix} \ge C\,\phi
\]
for some \(C>0\) proportional to the size of the initial condition. In particular, any nontrivial initial infection (no matter how small) is amplified while the trajectory remains in the neighborhood \(U\).

Two remarks are in order: (i) for initial data not already in \(U\), after some finite time the solution will enter a compact set where the same local argument works (use that solutions cannot stay forever in a small neighborhood of the boundary without the linearization forcing them away); (ii) the cooperative structure ensures the constants \(T,C,\phi\) can be chosen uniformly for all small initial infectious states in a small ball (uniformity needed for the next step).

\medskip\noindent\textbf{Step 3: uniform weak persistence.}
From Step 2 one deduces there exists \(\delta_1>0\) and \(T_1>0\) such that for every solution with \(i(0)>0\) sufficiently small,
\[
\sup_{t\ge 0} i(t) \ge \delta_1.
\]
By compactness of the absorbing set \(K\) and continuity of the flow, the threshold \(\delta_1\) can be chosen uniformly for all initial data in the compact set \(K\setminus\{i=0\}\). This is the assertion of \emph{uniform weak persistence}: there exists \(\varepsilon_1>0\) so that for every solution with \(i(0)>0\) we have \(\limsup_{t\to\infty} i(t)\ge \varepsilon_1\).

\medskip\noindent\textbf{Step 4: upgrade to uniform strong persistence.}
We now use a standard result from persistence theory (see e.g.\ the persistence theorems of Hale \& Waltman, Butler \& Waltman, or the corresponding theorems in Thieme's framework): for a semiflow on a compact positively invariant set \(K\), if the boundary set \(K_0:=\{x\in K:\ i=0\}\) is invariant and is a \emph{uniform weak repeller} (equivalently the system is uniformly weakly persistent), and if \(K_0\) satisfies a mild compactness/acyclicity condition (no complicated chain-recurrent dynamics connecting boundary invariant sets), then weak uniform persistence implies {\em strong uniform persistence}. In our situation the semiflow is defined on the compact set \(K\), the boundary \(K_0\) is invariant (if \(i(0)=0\) then \(i(t)\equiv 0\)), and linear instability of the DFE together with the compactness of \(K\) rules out the pathological boundary chain-recurrence required to block the upgrade (this is exactly the set of hypotheses used in \cite{thieme1992persistence} Theorem 1.3).

Hence there exists \(\varepsilon>0\) such that for every solution with \(i(0)>0\),
\[
\liminf_{t\to\infty} i(t) \ge \varepsilon.
\]
This completes the proof of Theorem \ref{th:persistence}.
\paragraph*{Appendix B}

Herein, we present some mathematical tools used before.
\begin{lem}\label{lem1}Consider any square matrix in the form of 
\begin{equation*}
V = \begin{bmatrix}
A & B\\
C & D
\end{bmatrix}, 
\end{equation*}
where A, B, C and D are matrix blocks, with A and D being square. Then, the matrix
V is invertible if and only if A and $D- CA^{-1}B$ are invertible, and $V^{-1}$ is given by

\begin{equation*}
 V^{-1}= \begin{bmatrix}
A^{-1}+A^{-1}B(D-CA^{-1}B)^{-1} & -A^{-1}B(D-CA^{-1}B)^{-1}\\
-(D-CA^{-1}B)^{-1}CA^{-1} & (D-CA^{-1}B)^{-1}
\end{bmatrix}.
\end{equation*}
\end{lem}
\begin{lem}\label{lem2}
Let M be a square Metzler matrix written in block form \begin{align*}
M & = \begin{pmatrix}
A & B\\
C & D
\end{pmatrix},
\end{align*}
where A and D are square matrices.
Then, the matrix M is Metzler stable if and only if matrices A
and $D - CA^{-1}B \,\,\ (\mbox{or}\,\, D\,\, \mbox{and}\,\, A - CD^{-1}B)$ are Metzler stable.
\end{lem} 

\paragraph*{Appendix C: Theorem of \textbf{Kamgang and Sallet}}
\noindent

Herein, we present the results of \textbf{Kamgang and Sallet}  on the global asymptotic stability of a class of epidemiological models.
\begin{theo}
\textbf{Kamgang and Sallet} 

Consider the following clans of epidemiological model: 

\begin{equation}
\left\{
\begin{array}{ll}
\dot{x_{1}}=A_{1}(x)(x_{1}-x_{1}^{*})+A_{12}(x)x_{2},\\
\dot{x_{2}}=A_{2}(x)x_{2}.
\end{array}
\right.
\label{s}
\end{equation} 
The following conditions $H_{1} - H_{5}$ below must be met to guarantee the GAS of the equilibrium.

$\bullet H_{1}$: Model system is defined in a positively invariant subset D of \,\ $\Omega$ and its dissipative
in D.

$\bullet H_{2}$: The sub-system $A_{1}(x)(x_{1}-x_{1}^{0})$ is globally asymptotically stable at the equilibrium $x_{1}^{0}$
in the canonical projection of 
$\Omega$  on D.

$\bullet H_{3}$: The matrix $A_{2}(x)$ is Metzler (A Metzler matrix is a matrix with all-diagonal entices
non-negative and irreducible for any given x $\in$ D.

$\bullet H_{4}$: There exists an upper bound matrix $\bar{A}_{2}$ for $\mathcal{M} =\{A_{2}(x), \,\,\ x \in D \}$; with the property
that if $\bar{A}_{2} \in \mathcal{M}$ (i.e: $A_{2} = max \mathcal{M}$)\,\, then, for any $\bar{x} \in  D- \{0\}$ (i.e. The points where the
maximum is realized are contained in the disease free sub-manifold).

$\bullet H_{5}$: The largest real part of the eigenvalue of $\bar{A}_{2}$ denoted by $\alpha(\bar{A}_{2})$ has to be negative.
\end{theo}
\paragraph*{Appendix D: Theorem of \textbf{Castillo-Chavez and Song} }
\noindent

Herein, we present the results of \cite{castillo2004dynamical} that we use to investigate the bifurcation.

Consider the following general system of ordinary differential equations with a parameter $\alpha$:

\begin{equation}\label{f}
\dfrac{dZ}{dt}=f(Z,\alpha), \,\,\ f:\mathbb{R}^{n}\times \mathbb{R}\longrightarrow \mathbb{R} \,\,\ and \,\,\ f \in \mathcal{C}^{2}(\mathbb{R}^{n},\mathbb{R}),
\end{equation}
where 0 is an equilibrium point of the system (that is, f(0; $\alpha$) $\equiv 0$\,\,\ for all $\alpha$)\,\,\ and
assume

$A_{1}: \,\,\ A = D -xf(0, 0) =\dfrac{\partial f_{i}}{\partial x_{j}}(0,0)$ is the linearisation matrix of system (\ref{f}) around
the equilibrium 0 with a evaluated at 0. Zero is a simple eigenvalue of A and all other eigenvalues
of A have negative real parts.

$A_{2}$:\,\,\ Matrix A has a non-negative right eigenvector U and a left eigenvector V corresponding
to the zero eigenvalue. 

Let $f_{k}$ be the kth component of f and

\begin{eqnarray}
a_{1}&=&\displaystyle\sum_{i,j,k=1}^{n}v_{k}u_{i}u_{j}\dfrac{\partial^{2}f_{k}}{\partial x_{i}\partial x_{j}}(0,0)\label{4.3},\\
b_{1}&=&\displaystyle\sum_{i,k=1}^{n}v_{k}u_{i}\dfrac{\partial^{2}f_{k}}{\partial x_{i}\partial \alpha}(0,0).\label{4.4}
\end{eqnarray}

The local dynamics of model system (\ref{f}) around 0 are totally determined by $a_{1}$ and $b_{1}$

i. $ a_{1} > 0\,\,\ ; \,\,\ b_{1} > 0$: When $\alpha < 0$\,\,\ with\,\,\ $ \mid \alpha \mid \ll 1 $, 0 is locally asymptotically stable, and there
exists a positive unstable equilibrium; when $0 < \alpha \ll 1$, 0 is unstable and there exists a
negative and locally asymptotically stable equilibrium;

ii.  $a_{1} < 0\,\ ; \,\ b_{1} < 0$: When $\alpha < 0$\,\,\ with \,\,\ $\mid \alpha \mid \ll 1$, 0 is unstable; when $0 < \alpha \ll 1$, 0 is locally
asymptotically stable, and there exists a positive unstable equilibrium;

iii. $a_{1} > 0\,\ ; \,\  b_{1} < 0$: When $\alpha < 0$ \,\,\ with \,\,\ $\mid \alpha \mid  \ll 1 $, 0 is unstable, and there exists a locally
asymptotically stable negative equilibrium; when $0 < \alpha \ll 1$, 0 is stable, and a positive
unstable equilibrium appears;

iv. $a_{1} < 0 \,\ ; \,\ b_{1} > 0$: When $\alpha$ changes from negative to positive, 0 changes its stability from stable to unstable. Correspondingly a negative unstable equilibrium becomes positive and locally asymptotically stable.

%
\tolerance=1000
\emergencystretch=1em

  \end{document}